\documentclass{amsart}

\usepackage{amssymb, amsmath, amscd, amsthm}
\usepackage{graphicx}
\usepackage{mathptmx}
\usepackage{psfrag}
\usepackage[all]{xy}

\newtheorem{theorem}{Theorem}[section]
\newtheorem{lem}[theorem]{Lemma}
\newtheorem{cor}[theorem]{Corollary}
\newtheorem{prop}[theorem]{Proposition}

\theoremstyle{definition}
\newtheorem{defi}[theorem]{Definition}

\newtheorem{exa}{\bf Example}
\theoremstyle{remark}

\newcommand{\F}{\mathcal F}
\newcommand{\R}{\mathbb R}
\newcommand{\N}{\mathbb N}
\newcommand{\Z}{\mathbb Z}

\newcommand{\p}{\varphi}
\newcommand{\s}{\psi}
\newcommand{\h}{h}
\newcommand{\f}{f}
\newcommand{\g}{g}
\newcommand{\rank}{\mbox{\rm{rank}}}

\newcommand{\eps}{\varepsilon}
\newcommand{\M}{{\mathcal M}}

\numberwithin{equation}{section}

\begin{document}

\title{Stability of Reeb graphs under function perturbations: the case of closed curves}

\author{B. Di Fabio}
\address{Dipartimento
di Matematica, Universit\`a di Bologna, P.zza di Porta S. Donato
5, I-$40126$ Bologna, Italia\newline ARCES, Universit\`a di
Bologna, via Toffano $2/2$, I-$40135$ Bologna, Italia}
\email{difabio@dm.unibo.it}

\author{C. Landi}
\address{Dipartimento
di Scienze e Metodi dell'Ingegneria, Universit\`a di Modena e
Reggio Emilia, Via Amendola 2, Pad. Morselli, I-42100 Reggio
Emilia, Italia\newline ARCES, Universit\`a di Bologna, via Toffano
$2/2$, I-$40135$ Bologna, Italia} \email{clandi@unimore.it}

\subjclass[2010]{Primary 68U05; Secondary 68T10; 05C10; 57R99}

\date{} 


\keywords{shape similarity,  editing distance, Morse function, natural stratification, natural pseudo-distance}

\begin{abstract}
Reeb graphs provide a method for studying the shape of a manifold
by encoding the evolution and arrangement of level sets of a
simple Morse function defined on the manifold. Since their
introduction in computer graphics they have been gaining
popularity as an effective tool for shape analysis and matching.
In this context one question deserving attention is whether Reeb
graphs are robust against function perturbations. Focusing on
1-dimensional manifolds, we define an editing distance between
Reeb graphs of curves, in terms of the cost necessary to transform
one graph into another. Our main result is that changes in Morse
functions induce smaller changes in the editing distance between
Reeb graphs of curves, implying stability of Reeb graphs under
function perturbations.
\end{abstract}

\maketitle

\section*{Introduction}
The shape similarity problem has since long been studied by the
computer vision community for dealing with shape classification
and retrieval tasks. It is now attracting more and more attention
also in the computer graphics community where recent improvements
in object acquisition and construction of digital models are
leading to an increasing accumulation of models in large databases
of shapes. Comparison of 2D images is often dealt with considering
just the silhouette or contour curve of the studied object,
encoding shape properties, such as curvature, in compact
representations of shapes, namely, shape descriptors, for the
comparison. The same approach is more and more used also in
computer graphics where there has been a gradual shift of research
interests from methods of representing shapes toward methods of
describing shapes of 3D models.

Since \cite{Sh91}, Reeb graphs have been gaining popularity as an
effective tool for shape analysis and description tasks   as a
consequence of their ability to extract high-level features from
3D models. Reeb graphs were originally defined by Georges Reeb in
1946 as topological constructs \cite{reeb46}. Given a manifold
$\M$ and a generic enough real-valued function $f$ defined on
$\M$, the simplicial complex defined by Reeb, conventionally
called the Reeb graph of $(\M,f)$, is the quotient space defined
by the equivalence relation that identifies the points of $\M$
belonging to the same connected component of level sets of $f$.
Reeb graphs effectively code shapes, both from a topological and a
geometrical perspective. While the topology is described by the
connectivity of the graph, the geometry can be coded in a variety
of different ways, according to the type of applications the Reeb
graph is devised for, simply by changing the function $f$.
Different choices of the function yield insights into the manifold
from different perspectives. The compactness of the
one-dimensional structure, the natural link between the function
and the shape, and the possibility of adopting different functions
for describing different aspects of shapes and imposing the
desired invariance properties,  have led to a great interest in
the use of Reeb graphs for similarity evaluation. In
\cite{HiSh*01}, Hilaga {\em et al.} use Multiresolution Reeb
Graphs based on the distribution of geodesic distance between two
points  as a search key for  3D objects, and the similarity
measure constructed in this setting is found to be resistant to
noise.  In this approach resistance to changes caused by noise
essentially relies on the choice of the geodesic distance to build
the Reeb graph. In \cite{BiMa*06}, Biasotti {\em et al.} base the
comparison of Extended Reeb Graphs on a relaxed version of the
notion of best common subgraph. This approach gives a method for
partial shape-matching able to recognize sub-parts  of objects,
and can be adapted to the context of applications since there is
no requirement on the choice of the function $f$.  Both
\cite{HiSh*01} and \cite{BiMa*06} present algorithms for
similarity evaluation.

To the best of our knowledge, mathematical assessment of stability
against function perturbations is still an open issue as far as
Reeb graphs are concerned. This question deserves attention since
it is clear that any data acquisition is subject to perturbations,
noise and approximation errors and, if Reeb graphs were not
stable, then distinct computational investigations of the same
object could produce completely different results. This paper aims
to be possibly the first positive answer to this question.

We confine ourselves to consider Reeb graphs of curves. In this
setting Reeb graphs are simply cycle graphs with an even number of
vertices corresponding alternatively to the maxima and minima of
the function. We also equip vertices of Reeb graphs with the value
taken by the function at the corresponding critical points.

Our main contribution is the construction of a distance between
Reeb graphs of curves such that changes in  functions imply
smaller changes in the  distance. Our distance is based on an
adaptation of the well-known notion of editing distance between
graphs \cite{Tai79}. We introduce three basic types of editing
operations, represented in Table \ref{deformations}, corresponding
to the insertion (birth)  of a new pair of adjacent points of
maximum and minimum, the deletion (death) of such a  pair, and the
relabelling of the vertices. A cost is associated with each of
these operations and our distance is given by the infimum of the
costs necessary to transform a graph into another by using these
editing operations. Our main result is the global stability of
labelled Reeb graphs under function perturbations (Theorem
\ref{global}):
\newline

\noindent\textsc{Main Result.} \emph{Let $f,g:S^1\to\R$ be two
simple Morse functions. Then the editing distance between the
labelled Reeb graph of $(S^1,f)$ and that of $(S^1,g)$ is always
smaller or equal to the $C^2$-norm of $f-g$.}
\newline

The main idea of the proof is to read editing operations in terms
of degenerate strata crossings of the space of smooth functions
stratified as in \cite{Cerf70}. We also obtain a lower bound for
our editing distance. Indeed, we find that it can be estimated
from below by the natural pseudo-distance between closed curves
studied in \cite{DoFr09}.

The paper is organized as follows. In Section \ref{background}, we
review some of the standard facts about Morse functions, the $C^r$
topology, the theory of stratification of smooth real valued
functions, and Reeb graphs. Section \ref{labreebgraphs} deals with
basic properties of labelled Reeb graphs of closed curves. Section
\ref{edit} is devoted to the definition of the admissible
deformations transforming a Reeb graph into another, the cost
associated with each kind of deformation, and the definition of an
editing distance in terms of this cost. Section \ref{lowbound} is
intended to provide a suitable lower bound for our distance, the
natural pseudo-distance; this represents a useful tool both to
show the well-definiteness of our distance and to compute it in
some simple cases. In Sections \ref{localstab} and
\ref{globalstab} it is shown that our distance is both locally and
globally upper bounded by the difference, measured in the
$C^2$-norm, between the functions defined on $S^1$. Eventually, a
brief discussion on the results obtained concludes the paper.

\section{Preliminary notions}\label{background}
In this section we recall some basic definitions and results about
Morse functions and Reeb graphs. Moreover, with the aim of proving
stability of Reeb graphs under function perturbations in mind, we
recall some concepts concerning the space of smooth real valued
functions on a smooth manifold: the $C^r$ topology and the theory
of the natural stratification.

Throughout the paper, $\M$ denotes a smooth (i.e. differentiable
of class $C^{\infty}$) compact $n$-manifold without boundary, and
$\F(\M,\R)$ the set of smooth real functions on $\M$.

\subsection{Simple Morse functions}
Let us recall the following concepts from \cite{Mi63}.

Let $f\in\F(\M,\R)$. A point $p\in \M$ is called a \emph{critical
point} of $f$ if, choosing a local coordinate system $(x_1,\ldots,
x_n)$ in a neighborhood $U$ of $p$, it holds that
$$\frac{\partial f}{\partial x_1}(p)=\ldots = \frac{\partial f}{\partial
x_n}(p)=0,$$ and it is called a \emph{regular point}, otherwise.
Throughout the paper, we set $K(f)=\{p\in\M: p \,\ \mbox{is a
critical point of}\,\ f\}$.

If $p\in K(f)$, then the real number $f(p)$ is called a
\emph{critical value} of $f$, and the set $\{q\in\M: q\in
f^{-1}(f(p))\}$ is called a \emph{critical level} of $f$.
Otherwise, if $p\notin K(f)$, then $f(p)$ is called a
\emph{regular value}. Moreover, a critical point $p$ is called
\emph{non-degenerate} if and only if the second derivative matrix
$$\left(\frac{\partial^2 f}{\partial x_i\partial x_j}(p)\right)$$
is non-singular, i.e. its determinant is not zero.

By the well-known Morse Lemma, in a neighborhood of a
non-degenerate critical point $p$, it is possible to choose a
local coordinate system $(x_1,\ldots, x_n)$ such that
$$f=f(p)-x_1^2-\ldots-x_{k}^2+x_{k+1}^2+\ldots+x_{n}^2.$$
The number $k$ is uniquely defined for each critical point $p$ and
is called the \emph{index} of $p$. Such an index completely
describes the behavior of $f$ at $p$. For example, $k=0$ means that the
corresponding $p$ is a minimum for $f$; $k=n$ means that $p$ is a
maximum; $0<k<n$ means that $p$ is a saddle point for $f$.

\begin{defi}
A function $f\in\F(\M,\R)$ is called a \emph{Morse function} if
all its critical points are non-degenerate. Moreover, a Morse
function is said to be \emph{simple} if each critical level
contains exactly one critical point.
\end{defi}

It is well-known that every Morse function has only finitely many
critical points (which are therefore certainly isolated points).
The importance of non-degeneracy is that it is the common
situation; indeed, in a sense that will be explained in Subsection
\ref{stratification}, the occurrence of degenerate critical points
is really quite rare.

\subsection{The $C^r$ topology on the space of real valued functions}
To topologize $\F(\M,\R)$, let us recall the definition of
$C^r$-norm, with $0\le r<\infty$ (see, e. g., \cite{Mi65,Pade82}).
Let $\{U_{\alpha}\}$ be a finite coordinate covering of $\M$, with
coordinate maps $h_{\alpha}: U_{\alpha}\to\R^n$, and consider a
compact refinement $\{C_{\alpha}\}$ of $\{U_{\alpha}\}$ (i.e.
$C_{\alpha}\subseteq U_{\alpha}$ for each $\alpha$, and $\bigcup
C_{\alpha}=\M$). For $f\in\F(\M,\R)$, let us set
$f_{\alpha}=f\circ h_{\alpha}^{-1}:h_{\alpha}(C_{\alpha})\to\R$.
Then the $C^r$-\emph{norm} of $f$ is defined as
$${\|f\|}_{C^r}=\underset{\alpha}\max \left\{\underset{u \in h_{\alpha}(C_{\alpha})}\max\left|f_{\alpha}(u)\right|, \underset{\underset{j\in\{1,\ldots,n\}}{u \in h_{\alpha}(C_{\alpha})}}\max\left|\frac{\partial f_{\alpha}}{\partial u_j}(u)\right|,\ldots, \underset{\underset{j_1,\ldots,j_r\in\{1,\ldots,n\}}{u \in h_{\alpha}(C_{\alpha})}}\max\left|\frac{{\partial}^r f_{\alpha}}{\partial u_{j_1}\cdots\partial u_{j_r}}(u)\right|\right\}.$$
The above norm defines a topology on $\F(\M,\R)$, known as the
$C^r$ \emph{topology} (or \emph{weak topology}), with $0\le
r<\infty$ (cf. \cite[chap. 2]{Hi76}). In the following, we will
denote by $B_r(f,\delta)$, $0\le r<\infty$, the open ball with
center $f$ and radius $\delta$ in the $C^r$ topology, i.e., $g\in
B_r(f,\delta)$ if and only if ${\|f-g\|}_{C^r}<\delta$. The
$C^{\infty}$ \emph{topology} is simply the union of the $C^r$
topologies on $\F(\M,\R)$ for every $0\le r<\infty$.

\subsection{Natural stratification of the space of real valued
functions}\label{stratification} Let us endow $\F(\M,\R)$ with the
$C^{\infty}$ topology, and consider the \emph{natural
stratification} of such a space, as exposed by Cerf in
\cite{Cerf70} (see also \cite{Se72}). The natural stratification
is defined as a sequence of sub-manifolds of $\F(\M,\R)$,
$\F^0,\F^1,\ldots,\F^j,\ldots$, of co-dimension
$0,1,\ldots,j,\ldots$, respectively, that constitute a partition
of $\F(\M,\R)$, and such that the disjoint union
$\F^0\cup\F^1\cup\ldots\cup\F^j$ is open for every $j$.

Before providing a brief description of the strata, let us recall
the following equivalence relation that can be defined on
$\F(\M,\R)$.
\begin{defi}
Two functions $f,g\in\F(\M,\R)$ are called \emph{topologically
equivalent} if there exists a diffeomorphism $\xi:\M\to\M$ and an
orientation preserving diffeomorphism $\eta:\R\to\R$ such that
$g(\xi(p))=\eta(f(p))$ for every $p\in\M$.
\end{defi}
\noindent The above relation is also known as \emph{isotopy} in
\cite{Cerf70}, and \emph{left-right equivalence} in
\cite{ArVaGu85}.

Let us describe $\F^0$ and $\F^1$, pointing out their main
properties that allow us to leave aside the remaining strata.

\begin{itemize}
\item The stratum $\F^0$ is the set of simple Morse functions.
\item The stratum $\F^1$ is the disjoint union of two sets
$\F^1_{\alpha}$ and $\F^1_{\beta}$ open in $\F^1$, where
\begin{itemize}
\item $\F^1_{\alpha}$ is the set of functions whose critical
levels contain exactly one critical point, and the critical points
are all non-degenerate, except exactly one. In a neighborhood of
such a point, say $p$, a local coordinate system $(x_1,\ldots,
x_n)$ can be chosen such that
$$f=f(p)-x_1^2-\ldots-x_{k}^2+x_{k+1}^2+\ldots+x_{n-1}^2+x_{n}^3.$$
\item $\F^1_{\beta}$ is the set of Morse functions whose critical
levels contain at most one critical point, except for one level
containing exactly two critical points.
\end{itemize}

\end{itemize}

$\F^0$ is dense in the space $\F(\M,\R)$ endowed with the $C^r$
topology, $2\le r\le\infty$ (cf. \cite[chap. 6, Thm. 1.2]{Hi76}).
Therefore, any smooth function can be turned into a simple Morse
function by arbitrarily small perturbations. Degenerate critical
points can be split into several non-degenerate singularities,
with all different critical values (Figure \ref{morseperturb}
$(a)$). Moreover, when more than one critical points occur at the
same level, they can be moved to close but different levels
(Figure \ref{morseperturb} $(b)$).

\begin{figure}[htbp]
\begin{center}
\psfrag{0}{$0$}\psfrag{v}{$\overline{\varepsilon}>0$}\psfrag{u}{$\overline{\varepsilon}<0$}\psfrag{e}{$\varepsilon$}
\psfrag{v5}{$v_5$}\psfrag{v6}{$v_6$}\psfrag{v7}{$v_7$}\psfrag{v8}{$v_8$}
\psfrag{f}{$f$}\psfrag{f3}{$f$}\psfrag{f1}{$\widetilde{f}_1$}\psfrag{f2}{$\widetilde{f}_2$}\psfrag{p1}{$p'$}\psfrag{p2}{$p''$}
\psfrag{p}{$p$}\psfrag{q}{$q$}\psfrag{q1}{$q'$}\psfrag{q2}{$q''$}
\psfrag{fi}{$\widetilde{f}_1$}\psfrag{fj}{$\widetilde{f}_2$}\psfrag{p3}{$\widetilde{p'}$}\psfrag{p4}{$\widetilde{p''}$}
\psfrag{c3}{$c''$}\psfrag{c2}{$c'$}\psfrag{c}{$c$}
\begin{tabular}{c}
\includegraphics[height=4cm]{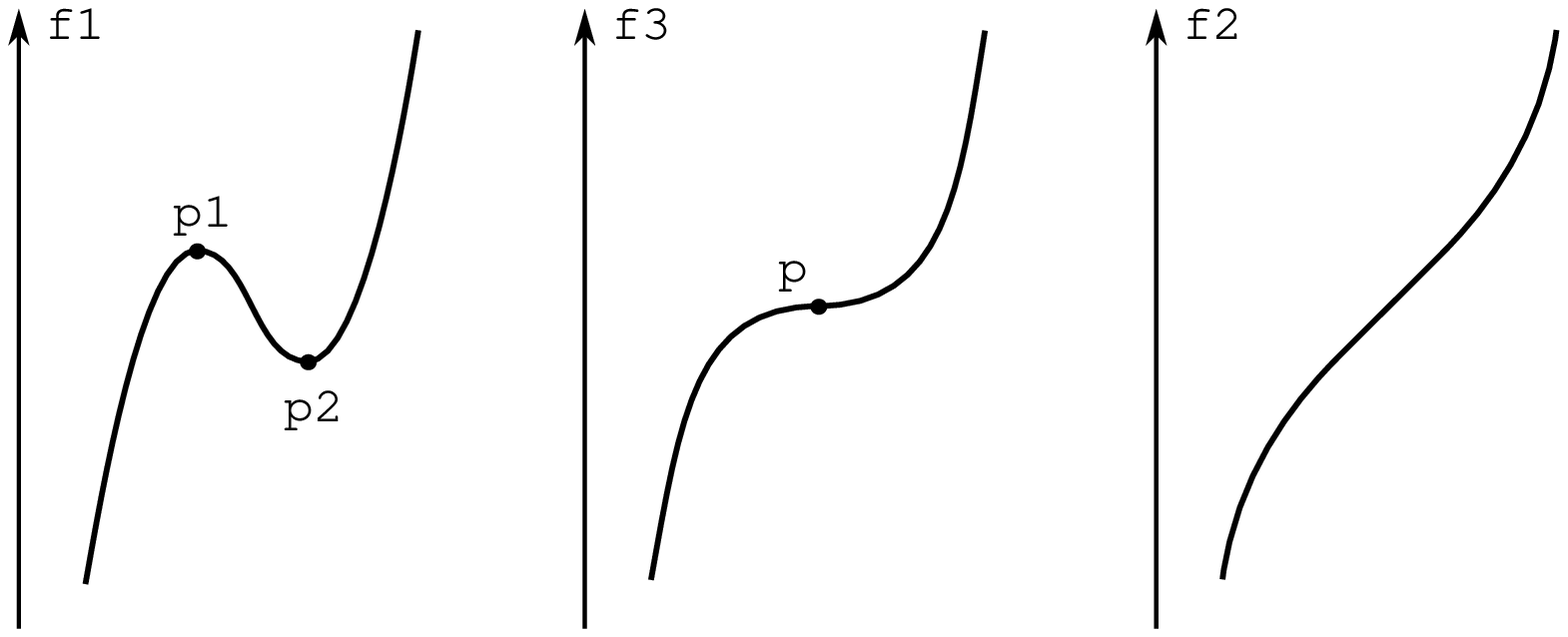}\\
$(a)$\\
\vspace{1mm}\\
\includegraphics[height=4cm]{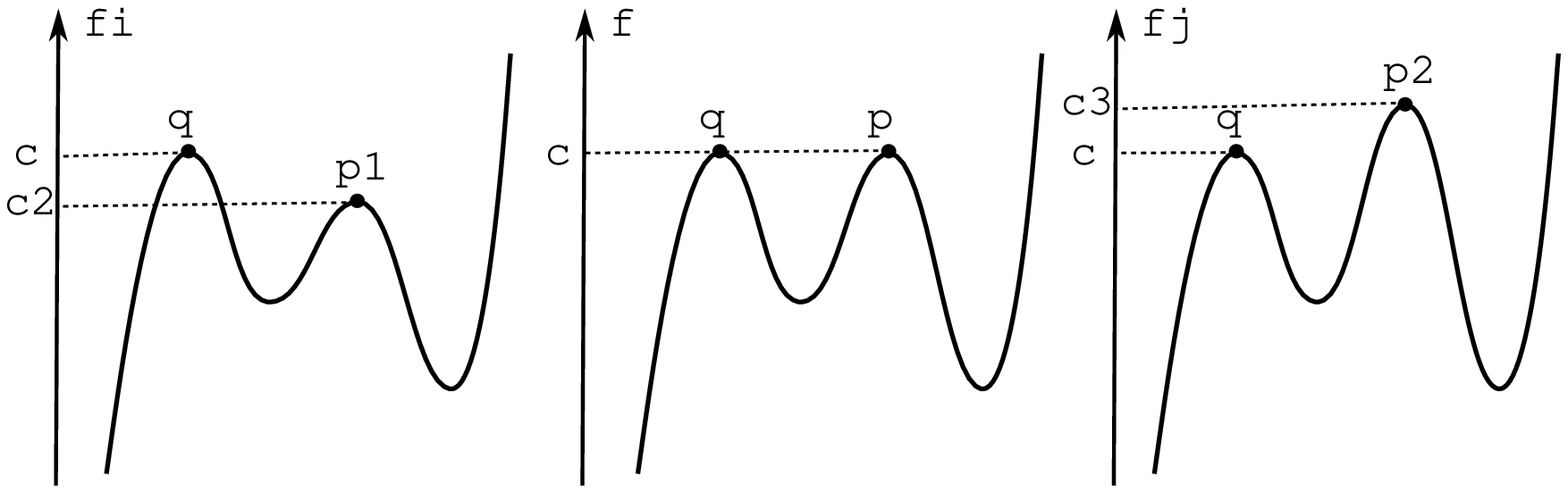}\\
$(b)$\\
\end{tabular}
\caption{\footnotesize{$(a)$ A function $f\in\F^1_{\alpha}$
admitting a degenerate critical point $p$ (center) can be
perturbed into a simple Morse function $\widetilde{f}_1$ with two
non-degenerate critical points $p',p''$ (left), or into a simple
Morse function $\widetilde{f}_2$ without critical points around
$p$ (right); $(b)$ a function $f\in\F^1_{\beta}$ (center) can be
turned into two simple Morse functions $\widetilde{f}_1,
\widetilde{f}_2$, that are not topologically equivalent
(left-right).}}\label{morseperturb}
\end{center} \end{figure}

It is well-known that two simple Morse functions are topologically
equivalent if and only if they belong to the same arcwise
connected component (or \emph{co-cellule}) of $\F^0$ \cite[p.
25]{Cerf70}.

$\F^1$ is a sub-manifold of co-dimension 1 of $\F^0\cup\F^1$, and
the complement of $\F^0\cup\F^1$ in $\F$ is of co-dimension
greater than 1. Consequently, given two functions $f,g\in\F^0$, we
can always find $\widehat{f},\widehat{g}\in\F(\M,\R)$ arbitrarily
near to $f, g$, respectively, for which the path
$h(\lambda)=(1-\lambda)\widehat{f}+\lambda \widehat{g}$, with
$\lambda\in [0,1]$, is such that
\begin{enumerate}
\item $\widehat{f},\widehat{g}\in\F^0$, and $\widehat{f}$, $\widehat{g}$ are topologically equivalent to $f$,
$g$, respectively; \item $h(\lambda)$ belongs to $\F^0\cup\F^1$
for every $\lambda\in [0,1]$; \item $h(\lambda)$ is transversal to
$\F^1$.
\end{enumerate}
As a consequence, $h(\lambda)$ belongs to $\F^1$ for at most a
finite collection of values $\lambda$, and does not traverse
strata of co-dimension greater than 1 (see, e.g., \cite{EdHa02}).

\subsection{Reeb graph of a manifold}\label{reebsection}
In this subsection we restate the main results concerning Reeb
graphs, starting from the following one shown by Reeb in
\cite{reeb46}. Here we consider pairs $(\M,f)$, with $\M$
connected and $f\in\F^0\subset\F(\M,\R)$.

\begin{theorem}\label{reeb}
The quotient space of $\M$ under the equivalence relation ``$p$
and $q$ belong to the same connected component of the same level
set of $f$'' is a finite and connected simplicial complex of
dimension 1.
\end{theorem}
This simplicial complex, denoted by $\Gamma_f$, is called the
\emph{Reeb graph} associated with the pair $(\M,f)$. Its vertex
set will be denoted by $V(\Gamma_f)$, and its edge set by
$E(\Gamma_f)$. Moreover, if $v_1,v_2\in V(\Gamma_f)$ are adjacent
vertices, i.e., connected by an edge, we will write $e(v_1,v_2)\in
E(\Gamma_{f})$. Since the vertices of a Reeb graph correspond in a
one to one manner to critical points of $f$ on the manifold $\M$
(see, e.g., \cite[Lemma 2.1]{BoFo04}), we will often identify each
$v\in V(\Gamma_f)$ with the corresponding $p\in K(f)$.

Given two topologically equivalent functions $f,g\in\F^0$, it is
well-known that the associated Reeb graphs, $\Gamma_{f}$ and
$\Gamma_{g}$, are isomorphic graphs, i.e., there exists an
edge-preserving bijection $\Phi: V(\Gamma_{f}) \rightarrow
V(\Gamma_{g})$. Beyond that, an even stronger result holds. Two
functions $f,g\in\F^0$ are topologically equivalent if and only if
such a bijection $\Phi$ also preserves the vertices order, i.e.,
for every $v,w\in V(\Gamma_f)$, $f(v)<f(w)$ if and only if
$g(\Phi(v))< g(\Phi(w))$.

The preceding result has been used  by Arnold in \cite{Ar07} to
classify simple Morse functions up to the topological equivalence
relation.

\section{Labelled Reeb graphs of closed curves}\label{labreebgraphs}

This paper focuses on \emph{Reeb graphs} of closed curves. Hence,
the manifold $\M$ that will be considered from now on is $S^1$,
and the function $f$ will be taken in $\F^0\subset\F(S^1,\R)$. The
Reeb graph $\Gamma_f$ associated with $(S^1,f)$ is a cycle graph
on an even number of vertices, corresponding, alternatively, to
the minima and maxima of $f$ on $S^1$ \cite{Pade82} (see, for
example, Figure \ref{labelledReeb} $(a)-(b)$). Furthermore, we
label the vertices of $\Gamma_f$, by equipping each of them with
the value of $f$ at the corresponding critical point. We denote
such a labelled graph by $(\Gamma_f, f_{_|})$, where
$f_{_|}:V(\Gamma_f)\to\R$ is the restriction of $f:S^1\to\R$ to
$K(f)$. A simple example is displayed in Figure \ref{labelledReeb}
$(a)-(c)$. To facilitate the reader, in all figures of this paper
we shall adopt the convention of representing $f$ as the height
function, so that $f_{_|}(v_a)<f_{_|}(v_b)$ if and only if $v_a$
is lower than $v_b$ in the picture.

\begin{figure}[htbp]
\psfrag{S}{$(S^1,f)$}\psfrag{G}{$\Gamma_f$}\psfrag{L}{$(\Gamma_f,
f_{_|})$}\psfrag{v1}{$v_1$}\psfrag{v2}{$v_2$}\psfrag{v3}{$v_3$}\psfrag{v4}{$v_4$}
\psfrag{v5}{$v_5$}\psfrag{v6}{$v_6$}\psfrag{v7}{$v_7$}\psfrag{v8}{$v_8$}\psfrag{f}{$f$}
\psfrag{a}{$(a)$}\psfrag{b}{$(b)$}\psfrag{c}{$(c)$}
\includegraphics[height=4cm]{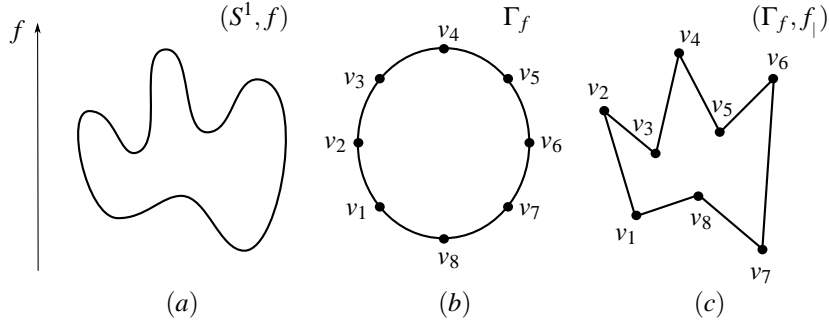}
\caption{\footnotesize{$(a)$ A pair $(S^1,f)$, with $f$ the height
function; $(b)$ the Reeb graph $\Gamma_f$ associated with
$(S^1,f)$; $(c)$ the labelled Reeb graph $(\Gamma_f, f_{_|})$
associated with $(S^1,f)$. Here labels are represented by the
heights of the vertices.}}\label{labelledReeb}
\end{figure}

The natural definition of isomorphism between labelled Reeb graphs
is the following one.

\begin{defi}\label{isolabel}
We shall say that two labelled Reeb graphs $(\Gamma_{f}, f_{_|}),
(\Gamma_{g}, g_{_|})$ are \emph{isomorphic} if there exists an
edge-preserving bijection $\Phi: V(\Gamma_{f}) \rightarrow
V(\Gamma_{g})$ such that $f_{_|}(v) = g_{_|}(\Phi(v))$ for every
$v\in V(\Gamma_{f})$.
\end{defi}

The following Proposition \ref{unique} provides a necessary and
sufficient condition in order that two labelled Reeb graphs are
isomorphic. It is based on the next definition of
re-parameterization equivalent functions.

\begin{defi}
Let $\mathcal H(S^1)$ be the set of homeomorphisms on $S^1$.
 We shall say that two functions $f,g\in\F^0\subset\F(S^1,\R)$ are
 \emph{re-parameterization
equivalent} if there exists $\tau\in\mathcal H(S^1)$ such that
$f(p)=g(\tau(p))$ for every $p\in S^1$.
\end{defi}

\begin{lem}\label{piecelinear}
Let $(\Gamma_{f},f_{_|})$ and $(\Gamma_{g},g_{_|})$ be labelled
Reeb graphs associated with $(S^1,f)$ and $(S^1,g)$, respectively.
If an edge-preserving bijection $\Phi: V(\Gamma_{f})\to
V(\Gamma_{g})$ exists, then there also exists a piecewise linear
$\tau \in \mathcal H(S^1)$ such that
$\tau_{|_{V(\Gamma_{f})}}=\Phi$. If moreover $f_{_|} =
g_{_|}\circ\Phi$, then $f = g\circ\tau$.
\end{lem}
\begin{proof}
The proof of the first statement is inspired by \cite[Lemma
4.2]{DoFr09}. Let us construct $\tau$ by extending $\Phi$ to $S^1$
as follows. Let us recall that $V(\Gamma_{f})=K(f)$ and
$V(\Gamma_{g})=K(g)$, and, by abuse of notation, for every pair of
adjacent vertices $p',p''\in V(\Gamma_f)$, let us identify the
edge $e(p',p'')\in E(\Gamma_f)$ with the arc of $S^1$ having
endpoints $p'$ and $p''$, and not containing any other critical
point of $f$. For every $p \in K(f)$, let $\tau(p)=\Phi(p)$. Now,
let us define $\tau(p)$ for every $p\in S^1\setminus K(f)$. Given
$p\in S^1\setminus K(f)$, we observe that there always exist
$p',p''\in V(\Gamma_f)$ such that $p\in e(p',p'')$. Since $\Phi$
is edge-preserving, there exists
$e(\Phi(p'),\Phi(p''))=e(\tau(p'),\tau(p''))\in E(\Gamma_g)$.
Hence, we can define $\tau(p)$ as the unique point of
$e(\tau(p'),\tau(p''))$ such that, if $f(p)= (1-\lambda_p)f(p')+
\lambda_p f(p'')$, with $\lambda_p\in [0,1],$ then $g(\tau(p))=
(1-\lambda_p)g(\tau(p'))+ \lambda_p g(\tau(p''))$. Clearly, $\tau$
belongs to ${\mathcal H}(S^1)$ and is piecewise linear.

As for the second statement, it is sufficient to observe that, if
$f_{_|} = g_{_|}\circ\Phi$, since $\tau(p)=\Phi(p)$ for every
$p\in K(f)$, then clearly $f_{_|}(p)=g_{_|}(\tau(p))$ for every
$p\in K(f)$. Moreover, for every $p\in S^1\setminus K(f)$, by the
construction of $\tau$, it holds that $g(\tau(p))=
(1-\lambda_p)g(\Phi(p'))+ \lambda_p
g(\Phi(p''))=(1-\lambda_p)f(p')+ \lambda_p f(p'')=f(p).$ In
conclusion, $f(p)=g(\tau(p))$ for every $p\in S^1$, and, hence,
$f,g$ are re-parameterization equivalent.

\end{proof}

\begin{prop}[Uniqueness theorem]\label{unique}
Let $(\Gamma_{f},f_{_|})$, $(\Gamma_{g},g_{_|})$ be labelled Reeb
graphs associated with $(S^1,f)$ and $(S^1,g)$, respectively. Then
$(\Gamma_{f},f_{_|})$ is isomorphic to $(\Gamma_{g},g_{_|})$ if
and only if $f$ and $g$ are re-parameterization equivalent.
\end{prop}
\begin{proof}
The direct statement is a trivial consequence of Lemma
\ref{piecelinear}.

As for the converse statement, it is sufficient to observe that
any $\tau\in\mathcal{H}(S^1)$ such that $f=g\circ\tau$, as well as
its inverse $\tau^{-1}$, takes the minima of $f$ to the minima of
$g$ and the maxima of $f$ to the maxima of $g$. Hence,
$\Phi:V(\Gamma_{f})\to V(\Gamma_{g})$, with
$\Phi=\tau_{|_{V(\Gamma_{f})}}$, is an edge preserving bijection
such that $f_{_|}=g_{_|} \circ\Phi.$
\end{proof}

As a consequence of Proposition \ref{unique}, two labelled Reeb
graphs isomorphic in the sense of Definition \ref{isolabel} will
always be identified, and in such case we will simply write
$(\Gamma_{f}, f_{_|})=(\Gamma_{g}, g_{_|})$.

The following Proposition \ref{realizth} ensures that, for every
cycle graph with an appropriate vertices labelling, there exists a
unique (up to re-parameterization) pair $(S^1, f)$, with
$f\in\F^0$, having such a graph as the associated labelled Reeb
graph.

\begin{prop}[Realization theorem]\label{realizth}
Let $(G,\ell)$ be a labelled graph, where $G$ is a cycle graph on
an even number of vertices, and $\ell:V(G)\rightarrow \R$ is an
injective function such that, for any vertex $v_2$ adjacent (that
is connected by an edge) to the vertices $v_1$ and $v_3$, either
both $\ell(v_1)$ and $\ell(v_3)$ are smaller than $\ell(v_2)$, or
both $\ell(v_1)$ and $\ell(v_3)$ are greater than $\ell(v_2)$.
Then there exists a simple Morse function $f:S^1\to \R$ such that
$(\Gamma_{f},f_{_{|K(f)}})=(G,\ell)$.
\end{prop}

\begin{proof}
It is evident.
\end{proof}

By virtue of the above Uniqueness and Realization theorems
(Propositions \ref{unique} and \ref{realizth}), for conciseness,
when a labelled Reeb graph will be introduced in the sequel, the
associated pair will be often omitted.

\section{Editing distance between labelled Reeb graphs}\label{edit}

We now define the editing deformations admissible to transform a
labelled Reeb graph of a closed curve into another. We introduce
at first elementary deformations and then the deformations
obtained by their composition. Next, we associate a cost with each
type of deformation, and define a distance between labelled Reeb
graphs in terms of such a cost.

\begin{defi}\label{elem_def}
Let  $(\Gamma_{f},f_{_|})$ be a labelled Reeb graph with $2n$
vertices, $n\geq 1$. We call an \emph{elementary deformation} of
$(\Gamma_{f},f_{_|})$ any of the following transformations:

\begin{itemize}
\item [(B)] (Birth): Assume $e(v_1, v_2) \in E(\Gamma_{f})$  with
$f_{_|}(v_{1})<f_{_|}(v_{2})$. Then $(\Gamma_{f},f_{_|})$ is
transformed into a labelled graph $(G,\ell)$ according to the
following rule: $G$ is the new graph on $2n+2$ vertices, obtained
deleting the edge $e(v_1, v_2)$ and inserting two new vertices
$u_{1}$, $u_{2}$ and the edges $e(v_1, u_1), e(u_1, u_2), e(u_2,
v_2)$; moreover, $\ell:V(G)\to\R$ is defined by extending $f_{_|}$
from $V(\Gamma_{f})$ to $V(G)=V(\Gamma_{f})\cup \{u_{1},u_{2}\}$
in such a way that $\ell_{|V(\Gamma_{f})}\equiv f_{_|}$, and
$f_{_|}(v_1)<\ell(u_2)<\ell(u_1)<f_{_|}(v_2)$.

\item [(D)] (Death): Assume $n\ge 2$, and $e(v_1, u_1), e(u_1,
u_2), e(u_2, v_2) \in E(\Gamma_{f})$, with
$f_{_|}(v_1)<f_{_|}(u_2)<f_{_|}(u_1)<f_{_|}(v_2)$. Then
$(\Gamma_{f},f_{_|})$ is transformed into a labelled graph
$(G,\ell)$ according to the following rule: $G$ is the new graph
on $2n-2$ vertices, obtained deleting $u_1$, $u_2$ and the edges
$e(v_1, u_1)$, $e(u_1, u_2)$, $e(u_2, v_2)$, and inserting an edge
$e(v_1, v_2)$; moreover, $\ell:V(G)\to\R$ is defined as the
restriction of $f_{_|}$ to $V(\Gamma_{f})\setminus \{u_1,u_2\}$.

\item[(R)] (Relabelling): $(\Gamma_{f},f_{_|})$ is transformed
into a labelled graph $(G,\ell)$ according to the following rule:
$G=\Gamma_{f}$, and for any vertex $v_2$ adjacent to the vertices
$v_1$ and $v_3$ (possibly $v_1 \equiv v_3$ for $n=1$), if both
$f_{_|}(v_1)$ and $f_{_|}(v_3)$ are smaller (greater,
respectively) than $f_{_|}(v_2)$, then both $\ell(v_1)$ and
$\ell(v_3)$ are smaller (greater, respectively) than $\ell(v_2)$;
moreover, for every $v\neq w$, $\ell(v)\neq \ell(w)$.
\end{itemize}
We shall denote by $T(\Gamma_{f},f_{_|})$ the result of the
elementary deformation $T$ applied to $(\Gamma_{f},f_{_|})$.
\end{defi}

Table \ref{deformations} schematically  illustrates the elementary
deformations described in Definition \ref{elem_def}.
\begin{table}[htbp]
\begin{center}
\psfrag{v1}{$v_1$}\psfrag{u1}{$u_1$}\psfrag{u2}{$u_2$}\psfrag{v2}{$v_2$}
\psfrag{v3}{$v_3$}\psfrag{v4}{$v_4$}
\psfrag{v5}{$v_5$}\psfrag{v6}{$v_6$}\psfrag{v7}{$v_7$}\psfrag{v8}{$v_8$}
\psfrag{B}{(B)}\psfrag{D}{(D)}\psfrag{R}{(R)}
\begin{tabular}{|c|}
\hline\\
\includegraphics[height=3cm]{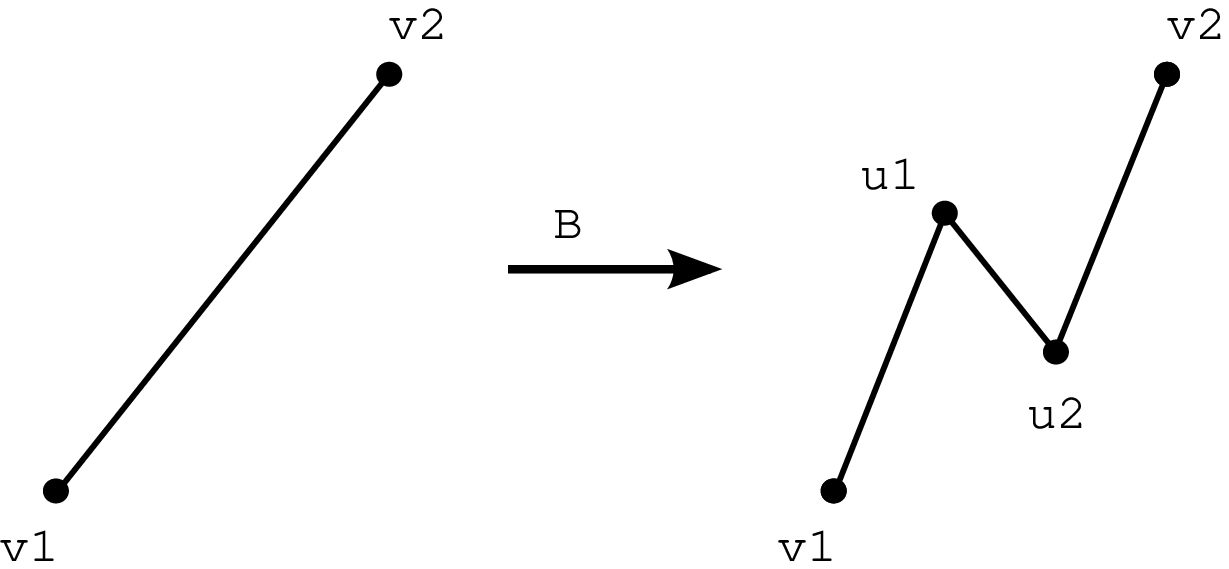}\\
\vspace{0mm}\\
\hline\\
\includegraphics[height=3cm]{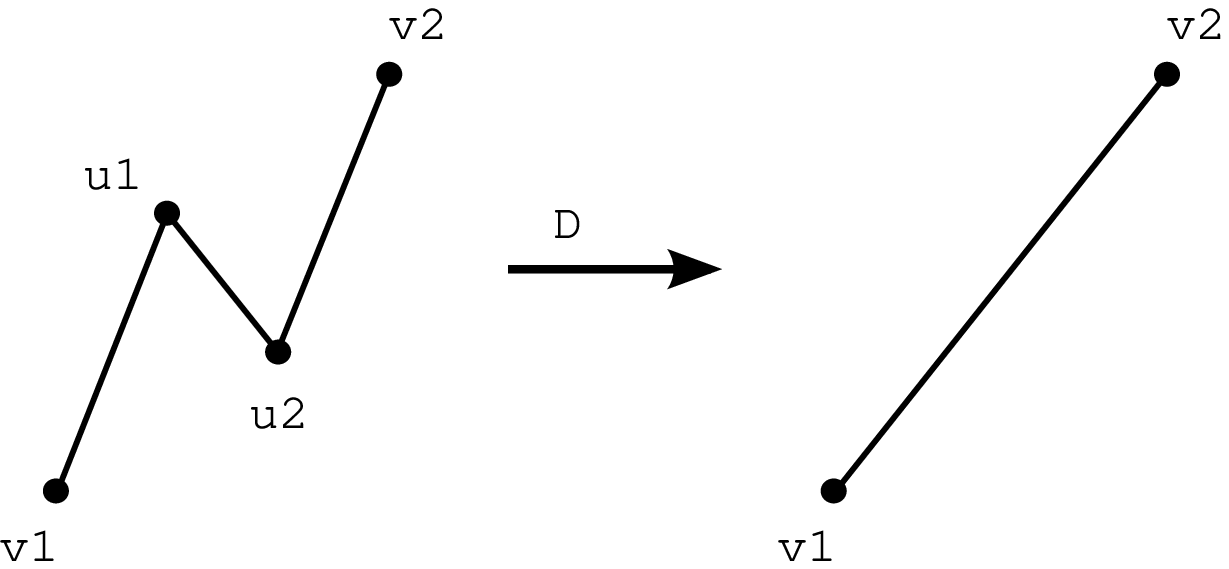}\\
\vspace{0mm}\\
\hline\\
\includegraphics[height=3cm]{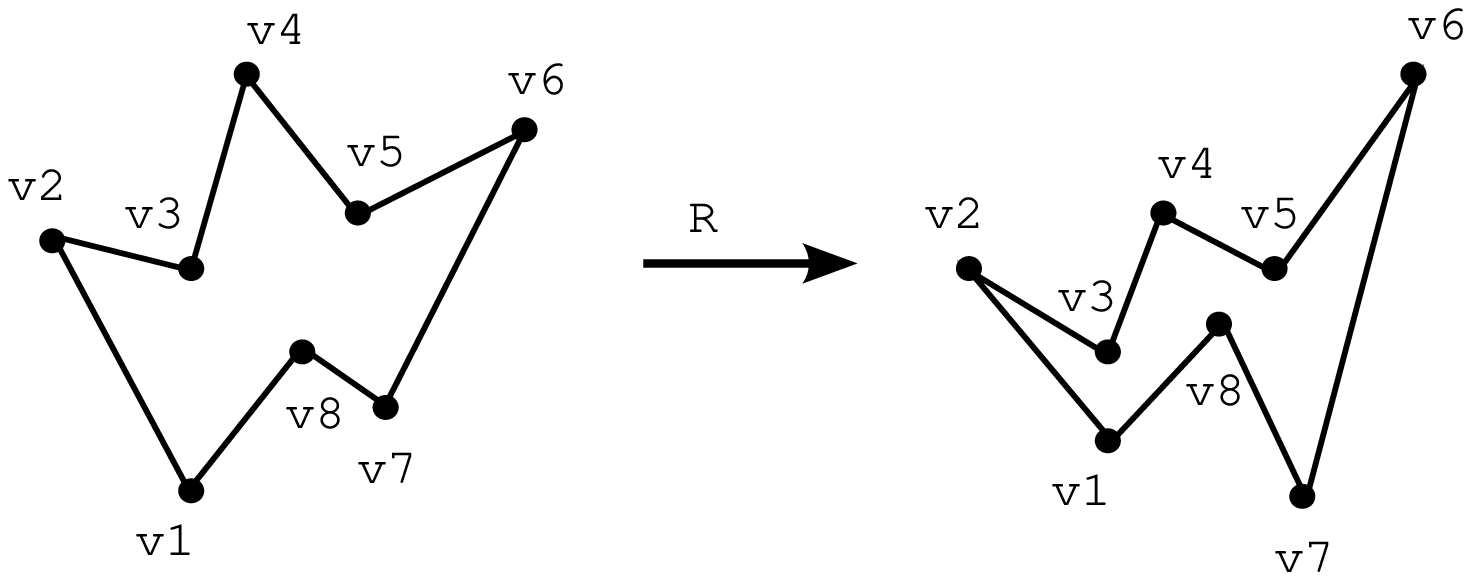}\\
\vspace{0mm}\\
\hline
\end{tabular}
\vspace{5mm} \caption{\footnotesize{The upper two figures
schematically show the elementary deformations of type (B) and
(D), respectively; the third figure shows an example of elementary
deformation of type (R).}}\label{deformations}
\end{center}
\end{table}

\begin{prop}\label{defGl}
Let $T$ be an elementary deformation of $(\Gamma_f,f_{_|})$, and
let $(G,\ell)=T(\Gamma_f,f_{_|})$. Then $(G,\ell)$ is a Reeb graph
$(\Gamma_g,g_{_|})$ associated with a pair $(S^1,g)$, and
$g\in\F^0$ is unique up to re-parameterization equivalence.
\end{prop}

\begin{proof}
The claim follows from Propositions \ref{realizth} and
\ref{unique}.
\end{proof}

\noindent As a consequence of the above result, from now on, we
will directly write $T(\Gamma_{f},f_{_|})=(\Gamma_{g},g_{_|})$.

Moreover, since the previous Proposition \ref{defGl} shows that an
elementary deformation of a labelled Reeb graph is still a
labelled Reeb graph, we can also apply elementary deformations
iteratively. This fact is used in the next Definition \ref{def}.

Given an elementary deformation $T$ of $(\Gamma_{f},f_{_|})$ and
an elementary deformation $S$ of $T(\Gamma_{f},f_{_|})$, the
juxtaposition $ST$ means applying first $T$ and then  $S$.

\begin{defi}\label{def}
We shall call {\em deformation} of $(\Gamma_{f},f_{_|})$ any
finite ordered sequence $T=(T_1,T_2,\ldots ,T_r)$ of elementary
deformations such that $T_1$ is an elementary deformation of
$(\Gamma_{f},f_{_|})$, $T_2$ is an elementary deformation of
$T_1(\Gamma_{f},f_{_|})$, ..., $T_r$ is an elementary deformation
of $T_{r-1}T_{r-2}\cdots T_1(\Gamma_{f},f_{_|})$. We shall denote
by $T(\Gamma_{f},f_{_|})$ the result of the deformation $T$
applied to  $(\Gamma_{f},f_{_|})$.
\end{defi}

Let us define the cost of a deformation.

\begin{defi}\label{cost 1}
Let $T$ be an elementary deformation transforming
$(\Gamma_{f},f_{_|})$ into $(\Gamma_{g},g_{_|})$.
\begin{itemize}
\item If $T$ is of type (B) inserting the vertices $u_1,u_2\in
V(\Gamma_{g})$, then we define the associated cost as
$$c(T)=\frac{|g_{_|}(u_1)-g_{_|}(u_2)|}{2};$$
\item If $T$ is of type (D) deleting the vertices $u_1,u_2\in
V(\Gamma_{f})$, then we define the associated cost as
$$c(T)=\frac{|f_{_|}(u_1)-f_{_|}(u_2)|}{2};$$
\item If $T$ is of type (R) relabelling the vertices $v\in
V(\Gamma_{f})=V(\Gamma_{g})$, then we define the associated cost
as
$$c(T)=\underset{v \in V(\Gamma_{f})}\max|f_{_|}(v)-g_{_|}(v)|.$$
\end{itemize}
Moreover, if $T=(T_1,\ldots ,T_r)$ is a deformation such that $T_r
\cdots T_1(\Gamma_{f}, f_{_|})=(\Gamma_{g}, g_{_|})$, we define
the associated cost as $c(T)=\underset{i=1}{\overset{r}\sum}
c(T_i)$.
\end{defi}

We now introduce the concept of inverse deformation.

\begin{defi}\label{definverse}
Let $T$ be a deformation such that $T(\Gamma_{f}, f_{_|}) =
(\Gamma_{g}, g_{_|})$. Then we denote by $T^{-1}$, and call it the
\emph{inverse} of $T$, the deformation such that
$T^{-1}(\Gamma_{g}, g_{_|})=(\Gamma_{f}, f_{_|})$ defined as
follows: \begin{itemize} \item If $T$ is elementary of type (B)
inserting two vertices, then $T^{-1}$ is of type (D) deleting the
same vertices; \item If $T$ is elementary of type (D) deleting two
vertices, then $T^{-1}$ is of type (B) inserting the same
vertices, with the same labels;\item If $T$ is elementary of type
(R) relabelling vertices of $V(\Gamma_{f})$, then $T^{-1}$ is
again of type (R) relabelling these vertices in the inverse way;
\item If $T=(T_1, \ldots, T_r)$, then $T^{-1}=(T^{-1}_r, \ldots,
T^{-1}_1)$.
\end{itemize}
\end{defi}

\begin{prop}\label{inverse}
For every deformation $T$ such that $T(\Gamma_{f}, f_{_|}) =
(\Gamma_{g}, g_{_|})$, $c(T^{-1})=c(T)$.
\end{prop}
\begin{proof}
Trivial.
\end{proof}

We prove that, for every two labelled Reeb graphs, a finite number
of elementary deformations always allows us to transform any of
them into the other one. We recall that we identify labelled Reeb
graphs that are isomorphic according to Definition \ref{isolabel}.
We first need a lemma, stating that in any labelled Reeb graph
with at least four vertices we can find two adjacent vertices that
can be deleted.

\begin{lem}\label{delete}
Let $(\Gamma_{f},f_{_|})$ be a labelled Reeb graph with at least
four vertices. Then there exist $e(v_1, u_1), e(u_1, u_2), e(u_2,
v_2) \in E(\Gamma_{f})$, with
$f_{_|}(v_1)<f_{_|}(u_2)<f_{_|}(u_1)<f_{_|}(v_2)$.
\end{lem}

\begin{proof}
Let $V(\Gamma_{f})=\{a_0,b_0,a_1,b_1,\ldots a_{m-1},b_{m-1}\}$,
$m\ge 2$. In the following, we convene that, for $k\in \Z$,
$a_{k}$ and $b_k$ are equal to $a_{(k\mod m)}$ and $b_{(k\mod
m)}$, respectively. We assume that
$E(\Gamma_{f})=\{e(a_i,b_i):i\ge 0\}\cup \{e(b_i,a_{i+1}):i\ge
0\}$, and $f_{_|}(a_i)<f_{_|}(b_i)$ for every $i$. From the
definition of labelled Reeb graph associated with a pair
$(S^1,f)$, it follows that $f_{_|}(b_i)>f_{_|}(a_{i+1})$,
$f_{_|}(a_i)\ne f_{_|}(a_{i+1})$, $f_{_|}(b_i)\ne
f_{_|}(b_{i+1})$, for every $i$.

The claim can be restated saying that there is at least one index $i$
such that either $(I)$ $f_{_|}(a_i)<f_{_|}(a_{i+1})$ and
$f_{_|}(b_i)<f_{_|}(b_{i+1})$ or $(II)$
$f_{_|}(a_{i+1})<f_{_|}(a_{i})$ and $f_{_|}(b_i)<f_{_|}(b_{i-1})$
hold. We prove this statement by contradiction, assuming that for
every $i\ge 0$ neither $(I)$ nor $(II)$ hold. Since $(I)$ does not
hold, either $f_{_|}(a_0)>f_{_|}(a_1)$ or
$f_{_|}(b_0)>f_{_|}(b_1)$ or both. Let us consider the case when
$f_{_|}(b_0)>f_{_|}(b_1)$. Since $(II)$ does not hold either, it
follows that $f_{_|}(a_{2})>f_{_|}(a_{1})$. Recalling that $(I)$
does not hold, we obtain $f_{_|}(b_1)>f_{_|}(b_{2})$. Iterating
the same argument, we deduce that $f_{_|}(b_i)>f_{_|}(b_i+1)$ for
every $i\ge 0$, contradicting the fact that $b_m=b_0$. An
analogous proof works when we consider the case
$f_{_|}(a_0)>f_{_|}(a_1)$.
\end{proof}

\begin{prop}\label{connected}
Let $(\Gamma_{f},f_{_|})$ and $(\Gamma_{g},g_{_|})$ be two
labelled Reeb graphs. Then the set of all the deformations $T$
such that $T(\Gamma_{f},f_{_|})=(\Gamma_{g},g_{_|})$ is non-empty.
This set of deformations will be denoted by ${\mathcal
T}((\Gamma_{f},f_{_|}), (\Gamma_{g},g_{_|})).$
\end{prop}

\begin{proof}
If $(\Gamma_{f},f_{_|})=(\Gamma_{g},g_{_|})$, then it is
sufficient to take the elementary deformation $T$ of type (R)
transforming $(\Gamma_{f},f_{_|})$ into itself. Otherwise, if
$(\Gamma_{f},f_{_|})\neq(\Gamma_{g},g_{_|})$ and $\Gamma_{f}$ has
at least four vertices, by Lemma \ref{delete}, we can apply a
finite sequence of elementary deformations of type (D) to
$(\Gamma_{f},f_{_|})$, so that in the resulting labelled Reeb
graph $(\Gamma_{h},h_{_|})$, $\Gamma_{h}$ has only two vertices,
say $u,v$, with $h_{_|}(u)<h_{_|}(v)$. If $(\Gamma_{g}, g_{_|})$
has also at least four vertices, by Lemma \ref{delete}, there
exists a finite sequence of elementary deformations of type (D) to
$(\Gamma_{g}, g_{_|})$, say $S=(S_1, \ldots, S_p)$, so that in the
resulting labelled Reeb graph $(\Gamma_{h'},h_{_|}')$,
$\Gamma_{h'}$ has only two vertices, say $u', v'$, with
$h_{_|}'(u')<h_{_|}'(v')$. So, we can apply to
$(\Gamma_{h},h_{_|})$ an elementary deformation of type (R) so to
obtain $(\Gamma_{h'},h_{_|}')$. Finally, by Definition
\ref{definverse}, we can apply to $(\Gamma_{h'},h_{_|}')$ the
finite sequence of elementary inverse deformations of type (B),
$S^{-1}=(S_p^{-1}, \ldots, S^{-1}_1)$, in order to obtain
$(\Gamma_{g},g_{_|})$. For $(\Gamma_{f}, f_{_|})$ or $(\Gamma_{g},
g_{_|})$ with only two vertices, the same proof applies without
need of deformations of type (D) or (B), respectively.
\end{proof}
A simple example explaining the above proof is given in Figure
\ref{Tnoempty}.

\begin{figure}[htbp]
\psfrag{v1}{$v_1$}\psfrag{v2}{$v_2$}\psfrag{v3}{$v_3$}\psfrag{v4}{$v_4$}
\psfrag{v5}{$v_5$}\psfrag{v6}{$v_6$}\psfrag{v7}{$v_7$}\psfrag{v8}{$v_8$}
\psfrag{u1}{$u_1$}\psfrag{u2}{$u_2$}\psfrag{u3}{$u_3$}\psfrag{u4}{$u_4$}
\psfrag{u5}{$u_5$}\psfrag{u6}{$u_6$}
\psfrag{B}{(B)}\psfrag{D}{(D)}\psfrag{R}{(R)}
\includegraphics[height=3cm]{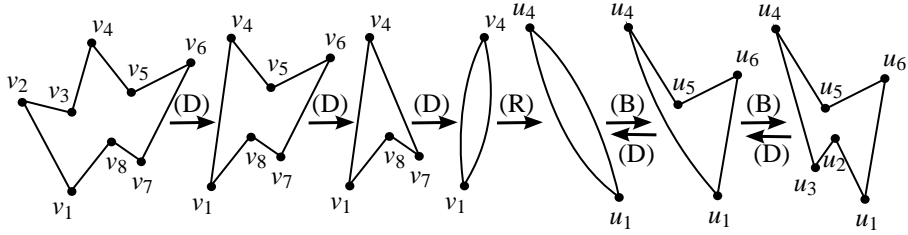}
\caption{\footnotesize{The leftmost labelled Reeb graph is
transformed into the rightmost one applying first three elementary
deformations of type (D), then one elementary deformation of type
(R), and finally two elementary deformations of type
(B).}}\label{Tnoempty}
\end{figure}

We point out that the deformation constructed in the proof of
Proposition \ref{connected} is not necessarily the cheapest one,
as can be seen in Example \ref{pse2}.

We now introduce an editing distance between labelled Reeb graphs,
in terms of the cost necessary to transform one graph into
another.

\begin{theorem}\label{editdist}
For every two labelled Reeb graphs $(\Gamma_{f}, f_{_|})$ and
$(\Gamma_{g},g_{_|})$, we set
$$d((\Gamma_{f}, f_{_|}),(\Gamma_{g},g_{_|}))=\inf_{T\in
\mathcal{T}((\Gamma_{f}, f_{_|}),(\Gamma_{g},g_{_|}))}c(T).$$ Then
$d$ is a distance.
\end{theorem}

The proof of the above theorem will be postponed to the end of the
following section. Indeed, even if the properties of symmetry and
triangular inequality can be easily verified, the property of the
positive definiteness of $d$ is not straightforward because the
set of all possible deformations transforming $(\Gamma_{f},
f_{_|})$ to $(\Gamma_{g},g_{_|})$ is not finite. In order to prove
the positive definiteness of $d$, we will need a further result
concerning the connection between the editing distance between two
labelled Reeb graphs, $(\Gamma_{f}, f_{_|})$,
$(\Gamma_{g},g_{_|})$, and the natural pseudo-distance between the
associated pairs $(S^1,f)$, $(S^1,g)$.

\section{A lower bound for the editing distance}\label{lowbound}
Now we provide a suitable lower bound for our editing distance by
means of the \emph{natural pseudo-distance}.

The natural pseudo-distance is a measure of the dissimilarity
between two pairs $(X,\p)$, $(Y,\s)$, with $X$ and $Y$ compact,
homeomorphic topological spaces and $\p:X\to\R$, $\s:Y\to\R$
continuous functions. Roughly speaking, it is defined as the
infimum of the variation of the values of $\p$ and $\s$, when we
move from $X$ to $Y$ through homeomorphisms (see \cite{DoFr04,
DoFr07, DoFr09} for more details).

Such a lower bound is useful for achieving two different results.
The first result, as mentioned in the preceding section, concerns
the proof of Theorem \ref{editdist}, i.e., that $d$ is a distance
(see Corollary \ref{defpos}). The second one is related to an
immediate question that can arise looking at the definition of
$d$: Is it always possible to effectively compute the cheapest
deformation transforming a labelled Reeb graph into another, since
the number of such deformations is not finite? By using the
natural pseudo-distance, we can estimate from below the value of
$d$, and, in certain simple cases, knowing the value of the
natural pseudo-distance allows us to determine the value of $d$
(see, e.g., Examples \ref{pse1}--\ref{pse2}).

The following Theorem \ref{lowerbound} states that the natural
pseudo-distance computed between the pairs $(S^1,f)$ and $(S^1,g)$
is a lower bound for the editing distance between the associated
labelled Reeb graphs.

\begin{theorem}\label{lowerbound}
Let $(\Gamma_{f},f_{_|})$, $(\Gamma_{g},g_{_|})$ be labelled Reeb
graphs associated with $(S^1,f)$ and $(S^1,g)$, respectively. Then
$d((\Gamma_{f},f_{_|}),(\Gamma_{g},g_{_|}))\ge \underset{\tau \in
{\mathcal H}(S^1)}\inf \|f-g\circ \tau\|_{C^0}.$
\end{theorem}

\begin{proof}
Let us prove that, for every $T \in {\mathcal
T}((\Gamma_{f},f_{_|}),(\Gamma_{g},g_{_|}))$, $c(T)\ge
\underset{\tau \in {\mathcal H}(S^1)}\inf\|f-g\circ \tau\|_{C^0}$.

First of all, assume that $T$ is an elementary deformation
transforming $(\Gamma_{f},f_{_|})$ into $(\Gamma_{g},g_{_|}).$ For
conciseness, slightly abusing notations, we will identify arcs of
$S^1$ having as endpoints two critical points $p',p''\in
V(\Gamma_f)$, and not containing other critical points of $f$,
with the edges $e(p',p'')\in E(\Gamma_f)$.
\begin{enumerate}
\item Let $T$ be of type (R) relabelling vertices of
$V(\Gamma_{f})$. Since, by Definition \ref{elem_def} (R),
$\Gamma_{f}=\Gamma_{g}$, we can always apply Lemma
\ref{piecelinear}, considering $\Phi$ as the identity map, to
obtain a piecewise linear $\tau\in {\mathcal H}(S^1)$ such that
$\tau(p)=p$ for every $p\in K(f)$. As far as non-critical points
are concerned, following the proof of Lemma \ref{piecelinear}, for
every $p\in S^1\setminus K(f)$, $\tau(p)$ is defined as that point
on $S^1$ such that, if $p\in e(p',p'')\in E(\Gamma_{f})$, with
$f(p)=(1-\lambda_p)f(p')+\lambda_p f(p'')$, $\lambda_p\in [0,1]$,
then $\tau(p)\in e(p',p'')$ with
$g(\tau(p))=(1-\lambda_p)g(p')+\lambda_p g(p'')$. Therefore, by
substituting to $f(p)$ and $g(\tau(p))$ the above expressions, we
see that $\underset{p \in S^1}\max|f(p)-g(\tau(p))|=\underset{p
\in V(\Gamma_{f})}\max|f_{_|}(p)-g_{_|}(p)|=c(T).$

\item Let $T$ be of type (D) deleting $q_1, q_2 \in
V(\Gamma_{f})$, the edges $e(p_1, q_1)$, $e(q_1, q_2)$, $e(q_2,
p_2)$, and inserting the edge $e(p_1, p_2)$. Thus, for every $p\in
K(f)\backslash \{q_1,q_2\}$, $f(p)=g(p)$. It is not restrictive to
assume that $f(p_1)<f(q_2)<f(q_1)<f(p_2)$. Then we can define a
sequence $(\tau_n)$ of piecewise linear homeomorphisms on $S^1$
approximating this elementary deformation. Let $\tau_n(p)=p$ for
every every $p \in V(\Gamma_{f})\backslash
\{q_1,q_2\}=V(\Gamma_{g})$ and $n\in \N$. Moreover, let
$\overline{q}$ be the point of $e(p_1,p_2)\in E(\Gamma_{g})$ such
that $\g(\overline{q})=\frac{\f(q_1)+\f(q_2)}{2}$ (such a point
$\overline{q}$ exists because
$g(p_1)=f(p_1)<f(q_2)<f(q_1)<f(p_2)=g(p_2)$ and it is unique
because we are assuming that no critical points of $g$ occur in
the considered arc). Let us fix a positive real number $c< \min
\{\g(p_2)-g(\overline{q}),\g(\overline{q})-\g(p_1)\}$. For every
$n\in\N$, let us define $\tau_n(q_1)$ (resp. $\tau_n(q_2)$) as the
only point on $S^1$ belonging to the arc with endpoints
$p_1,\overline{q}$ (resp. $\overline{q},p_2$) contained in
$e(p_1,p_2)$, such that $\g(\tau_n(q_1))=\g(\overline{q}) -
\frac{c}{n}$ (resp. $\g(\tau_n(q_2))=\g(\overline{q}) +
\frac{c}{n}$) as shown in Figure \ref{Delproof}. Now, let us
linearly extend $\tau_n$ to all $S^1$ in the following way. For
every $p\in S^1\setminus K(f)$, if $p$ belongs to the arc with
endpoints $p',p''\in K(f)$ not containing any other critical
point, and is such that $f(p)=(1-\lambda_p)f(p')+\lambda_p
f(p'')$, $\lambda_p\in [0,1]$, then $\tau_n(p)$ belongs to the arc
with endpoints $\tau_n(p'),\tau_n(p'')$ not containing any other
critical point, and is such that
$g(\tau_n(p))=(1-\lambda_p)g(\tau_n(p'))+\lambda_p
g(\tau_n(p''))$. Hence, $\tau_n$ is piecewise linear for every
$n\in\N$, and $\underset{n\to \infty}\lim\underset{p \in
S^1}\max|\f(p)-\g(\tau_n(p))|=\underset{n\to
\infty}\lim\underset{p \in
V(\Gamma_{f})}\max|\f(p)-\g(\tau_n(p))|=\underset{n\to
\infty}\lim\max\{\f(q_1)-\g(\tau_n(q_1)),\f(q_2)-\g(\tau_n(q_2))\}=|\f(q_1)-\g(\overline{q})|=\frac{\f_{_|}(q_1)-\f_{_|}(q_2)}{2}=c(T).$
\begin{figure}[htbp]
\psfrag{p1}{$p_1$}\psfrag{p2}{$p_2$}\psfrag{q1}{$q_1$}\psfrag{q2}{$q_2$}\psfrag{tq1}{$\tau_n(q_1)$}\psfrag{tq2}{$\tau_n(q_2)$}
\psfrag{q}{$\overline{q}$}\psfrag{tp1}{$\tau_n(p_1)=p_1$}\psfrag{tp2}{$\tau_n(p_2)=p_2$}
\psfrag{f1}{$\frac{\f(q_1)+\f(q_2)}{2}$}\psfrag{X}{$(S^1,f)$}\psfrag{Y}{$(S^1,g)$}\psfrag{u4}{$u_4$}
\psfrag{u5}{$u_5$}\psfrag{u6}{$u_6$}
\psfrag{B}{(B)}\psfrag{D}{(D)}\psfrag{R}{(R)}
\includegraphics[height=6cm]{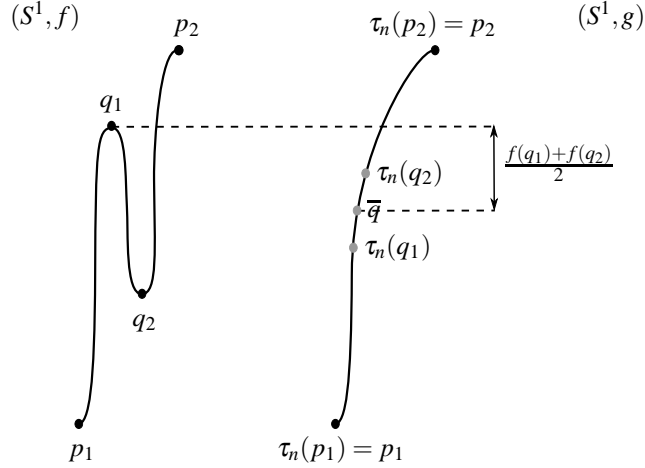}
\caption{\footnotesize{The construction of the homomorphism
$\tau_n$ as described in step (2) of the proof of Theorem
\ref{lowerbound}. The arc $e(p_1,q_1)$ ($e(q_1,q_2)$, and
$e(q_2,p_2)$, respectively) is piecewise linearly taken to the arc
having $\tau_n(p_1),\tau_n(q_1)$ ($\tau_n(q_1),\tau_n(q_2)$ and
$\tau_n(q_2),\tau_n(p_2)$, respectively) as
endpoints.}}\label{Delproof}
\end{figure}
\item Let $T$ be of type (B) deleting $e(p_1, p_2)\in
E(\Gamma_{f})$, and inserting two vertices $q_1, q_2$ and the
edges $e(p_1, q_1)$, $e(q_1, q_2)$, $e(q_2, p_2)$. Then we can
apply the same proof as (2), by considering the inverse
deformation $T^{-1}$ that, by Definition \ref{definverse}, is of
type (D) and,  by Proposition \ref{inverse}, has the same cost of
$T$.
\end{enumerate}
Therefore, observing that in (1), the piecewise linear $\tau$ can
be clearly replaced by a sequence $(\tau_n)$, with $\tau_n=\tau$
for every $n\in\N$, we can assert that, for every elementary
deformation $T$, there exists a sequence of piecewise linear
homeomorphisms on $S^1$, $(\tau_n)$, such that
$c(T)=\underset{n\to \infty}\lim {\|f-g\circ \tau_n\|}_{C^0}\ge
\underset{\tau \in {\mathcal H}(S^1)}\inf\|f-g\circ \tau\|_{C^0}$.

\noindent Now, let $T=(T_1, \ldots, T_r)\in {\mathcal
T}((\Gamma_{f},f_{_|}),(\Gamma_{g},g_{_|}))$ and prove that, also
in this case, $c(T)\ge \underset{\tau \in {\mathcal
H}(S^1)}\inf\|f-g\circ \tau\|_{C^0}$. Let us set $T_i \cdots
T_1(\Gamma_{f},f_{_|})=(\Gamma_{f^{(i)}},f^{(i)}_{_|})$,
$f=f^{(0)}$, $g=f^{(r)}$. For $i=1,\ldots,r$, let
${(\tau_n^{(i)})}_n$ be a sequence of piecewise linear
homeomorphisms on $S^1$ for which it holds that
$c(T_i)=\underset{n\to \infty}\lim {\|f^{(i-1)}-f^{(i)}\circ
\tau_n^{(i)}\|}_{C^0}$, and let ${(\tau_n^{(0)})}_n$ be the
constant sequence such that $\tau_n^{(0)}= Id$ for every $n\in
\N$. Then {\setlength\arraycolsep{2pt}\begin{eqnarray*} c(T)&=&
\underset{i=1}{\overset{r}\sum}c(T_i)= \underset{n\to \infty}\lim
\|\f^{(0)}-\f^{(1)}\circ \tau_n^{(1)}\|_{C^0}+
\underset{i=1}{\overset{r-1}\sum}\underset{n\to \infty}\lim
\|\f^{(i)}-\f^{(i+1)}\circ
\tau_n^{(i+1)}\|_{C^0}\\&=&\underset{n\to \infty}\lim
\|\f^{(0)}-\f^{(1)}\circ \tau_n^{(1)}\|_{C^0}\\&&+
\underset{i=1}{\overset{r-1}\sum}\underset{n\to \infty}\lim
\|\f^{(i)}\circ \tau_{n}^{(i)}\circ \ldots \circ
\tau_n^{(0)}-\f^{(i+1)}\circ \tau_n^{(i+1)}\circ
\tau_{n}^{(i)}\circ \cdots \circ \tau_n^{(0)}\|_{C^0}\\&\ge &
\underset{r\to \infty}\lim \|\f^{(0)}-\f^{(r)}\circ
\tau_n^{(r)}\circ \tau_{n}^{(r-1)}\circ \cdots \circ
\tau_n^{(0)}\|_{C^0}\ge \underset{\tau \in {\mathcal
H}(S^1)}\inf\|\f-\g\circ \tau\|_{C^0},
\end{eqnarray*}}
where the third equality is obtained by observing that
$$\f^{(i)}\circ \tau_{n}^{(i)}\circ \cdots \circ
\tau_n^{(0)}-\f^{(i+1)}\circ \tau_n^{(i+1)}\circ
\tau_{n}^{(i)}\circ \cdots \circ
\tau_n^{(0)}=(\f^{(i)}-\f^{(i+1)}\circ \tau_n^{(i+1)})\circ
\tau_{n}^{(i)}\circ \cdots \circ \tau_n^{(0)}$$ for every $i \in
\{1,\ldots,r-1\}$, and that ${\|\cdot\|}_{C^0}$ is invariant under
re-parameterization; the first inequality is consequent to the
triangular inequality.
\end{proof}
\begin{cor}\label{defpos}
If $d((\Gamma_{f},f_{_|}),(\Gamma_{g},g_{_|}))=0$ then
$(\Gamma_{f},f_{_|})=(\Gamma_{g},g_{_|})$.
\end{cor}
\begin{proof}
From Theorem \ref{lowerbound},
$d((\Gamma_{f},f_{_|}),(\Gamma_{g},g_{_|}))=0$ implies that
$\underset{\tau \in {\mathcal H}(S^1)}\inf \|f-g\circ
\tau\|_{C^0}=0$. In \cite{CeDi09} it has been proved that when
$\underset{\tau \in {\mathcal H}(X,Y)}\inf \|f-g\circ
\tau\|_{C^0}=0$, with $X$, $Y$ two closed curves of class at least
$C^2$, a homeomorphism $\overline{\tau}\in {\mathcal H}(X,Y)$
exists such that $\f=\g\circ \overline{\tau}.$ Therefore, the
claim follows from Proposition \ref{unique}.
\end{proof}
\begin{proof}[Proof of Theorem \ref{editdist}]
The positive definiteness of $d$ has been proved in Corollary
\ref{defpos}; the symmetry is a consequence of Proposition
\ref{inverse}; the triangular inequality can be easily verified in
the standard way.
\end{proof}

Now we describe two simple examples showing how it is possible to
compute the editing distance between two labelled Reeb graphs,
$(\Gamma_{f},f_{_|}),(\Gamma_{g},g_{_|})$, by exploiting the
knowledge of the natural pseudo-distance value between the
associated pairs $(S^1,f), (S^1,g)$. In particular, Example
\ref{pse1} provides a situation in which the infimum cost over all
the deformations transforming $(\Gamma_{f},f_{_|})$ into
$(\Gamma_{g},g_{_|})$ is actually a minimum. In Example \ref{pse2}
this infimum is obtained by applying a passage to the limit.

\begin{exa}\label{pse1}
Let us consider the two pairs $(S^1,f), (S^1,g)$ depicted in
Figure \ref{pseudo}, with $f,g\in\F^0$. We now show that
$d((\Gamma_{f},f_{_|}),(\Gamma_{g},g_{_|}))=
\frac{1}{2}(f(q_1)-f(p_1))$. Indeed, in this case, the natural
pseudo-distance between $(S^1,f)$ and $(S^1,g)$ is equal to
$\frac{1}{2}(f(q_1)-f(p_1))$ (cf. \cite{DoFr09}). Therefore, by
Theorem \ref{lowerbound}, it follows that
$d((\Gamma_{f},f_{_|}),(\Gamma_{g},g_{_|}))\geq
\frac{1}{2}(f(q_1)-f(p_1))$. On the other hand, the deformation
$T$ of type (D) that deletes the vertices $p_1,q_1\in
V(\Gamma_f)$, the edges $e(p,q_1),e(q_1,p_1),e(p_1,q)$ and inserts
the edge $e(p,q)$ transforms $(\Gamma_{f},f_{_|})$ into
$(\Gamma_{g},g_{_|})$ with cost $c(T)=\frac{1}{2}(f(q_1)-f(p_1))$.
Hence
$d((\Gamma_{f},f_{_|}),(\Gamma_{g},g_{_|}))=\frac{1}{2}(f(q_1)-f(p_1))$.

\begin{figure}[htbp]
\begin{center}
\psfrag{A}{$q_1$} \psfrag{B}{$p_1$} \psfrag{C}{$q$}
\psfrag{D}{$q'$} \psfrag{E}{$p$} \psfrag{F}{$p'$} \psfrag{z}{$z$}
\psfrag{e}{$2\eps$} \psfrag{X}{$(S^1,f)$} \psfrag{Y}{$(S^1,g)$}
\psfrag{M}{$\max z$} \psfrag{m}{$\min z$}
\includegraphics[height=5cm]{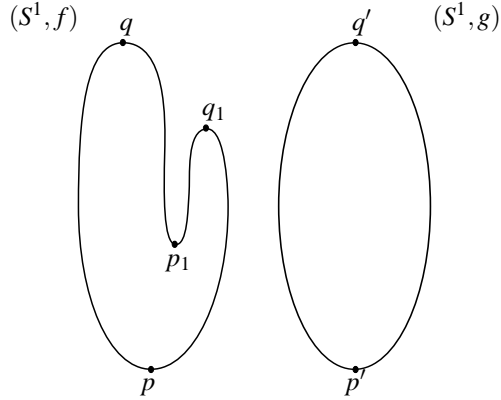}
\caption{\footnotesize{The pairs considered in Example \ref{pse1}.
In this case
$d((\Gamma_{f},f_{_|}),(\Gamma_{g},g_{_|}))=\underset{\tau \in
{\mathcal H}(S^1)}\inf\|f-g\circ
\tau\|_{C^0}=\frac{1}{2}(f(q_1)-f(p_1))$}}\label{pseudo}
\end{center}
\end{figure}
\end{exa}

\begin{exa}\label{pse2}
Let us consider now the two pairs $(S^1,f), (S^1,g)$ illustrated
in Figure \ref{pseudo2}. Let $f(q_1)-f(p_1)=f(q_2)-f(p_2)=a$.
Then, clearly, $\underset{\tau \in {\mathcal
H}(S^1)}\inf\|f-g\circ \tau\|_{C^0}=\frac{a}{2}$. Let us show that
the editing distance between $(\Gamma_{f},f_{_|})$ and
$(\Gamma_{g},g_{_|})$ is $\frac{a}{2}$, too.
\begin{figure}[htbp]
\begin{center}
\psfrag{A1}{$q_1$} \psfrag{B1}{$p_1$}\psfrag{A2}{$q_2$}
\psfrag{B2}{$p_2$} \psfrag{C}{$q$} \psfrag{D}{$q'$}
\psfrag{E}{$p$} \psfrag{F}{$p'$} \psfrag{z}{$z$}
\psfrag{e}{$2\eps$} \psfrag{X}{$(S^1,f)$} \psfrag{Y}{$(S^1,g)$}
\psfrag{M}{$\max z$} \psfrag{m}{$\min z$}
\includegraphics[height=5cm]{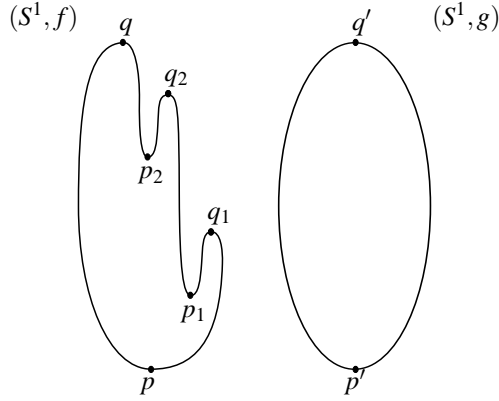}
\caption{\footnotesize{The pairs considered in Example \ref{pse2}.
Even in this case
$d((\Gamma_{f},f_{_|}),(\Gamma_{g},g_{_|}))=\underset{\tau \in
{\mathcal H}(S^1)}\inf\|f-g\circ
\tau\|_{C^0}=\frac{1}{2}(f(q_1)-f(p_1))$}}\label{pseudo2}
\end{center}
\end{figure}
For every $0<\epsilon<\frac{a}{2}$, we can apply to
$(\Gamma_{f},f_{_|})$ a deformation of type (R), that relabels
$p_1,p_2,q_1,q_2$ in such a way that $f(p_i)$ is increased of
$\frac{a}{2}-\epsilon$, and $f(q_i)$ is decreased of
$\frac{a}{2}-\epsilon$ for $i=1,2$, composed with two deformations
of type (D) that delete $p_i$ with $q_i$, $i=1,2$. Thus, since the
total cost is equal to $\frac{a}{2}-\epsilon + 2\epsilon$, by the
arbitrariness of $\epsilon$, it holds that
$d((\Gamma_{f},f_{_|}),(\Gamma_{g},g_{_|}))\leq\frac{a}{2}.$
Applying Theorem \ref{lowerbound}, we deduce that
$d((\Gamma_{f},f_{_|}),(\Gamma_{g},g_{_|}))=\frac{a}{2}.$
\end{exa}

\section{Local stability}\label{localstab}
This section is intended to show that labelled Reeb graphs of
closed curves are stable under small function perturbations with
respect to our editing distance (see Theorem \ref{local}). The
main tool we will use is provided by Theorem \ref{mainres}, that
ensures the stability of simple Morse function critical values.
This latter result can be deduced by the homological properties of
the lower level sets of a simple Morse function $f$ on a manifold
$\M$, and its validity does not depend on the dimension of $\M$.
Therefore, it will be given for any smooth compact manifold
without boundary.

For every $f\in \F(\M, \R)$, and for every $a\in\R$, let us denote
by $f^a$ the lower level set $f^{-1}(-\infty,a]=\{p\in \M:
f(p)\leq a\}$. Let us recall the existing link between the
topology of a pair of lower level sets $(f^b,f^a)$, with
$a,b\in\R$, $a<b$,  regular values of $f$, and the critical points
of $f$ lying between $a$ and $b$. The following statements hold
(cf. \cite{Mi63}):
\begin{itemize}
\item[\emph{(St.~1)}] If the interval $f^{-1}([a,b])$ contains no
critical points, then $f^a$ is a deformation retract of $f^b$, so
that the inclusion map $f^a\to f^b$ is a homotopy equivalence.

\item[\emph{(St.~2)}] If $f^{-1}([a,b])$ contains exactly one
critical point of index $\overline{k}$, then, denoting by $G$ the
homology coefficient group, it holds that
$$H_k(f^b,f^a)=\left\{
\begin{array}{ll}
G,&\mbox{if}\,\ k=\overline{k}\\0,&\mbox{otherwise}.
\end{array}\right.
$$
\end{itemize}
In the remainder of this section we require $f$ to be a simple
Morse function. Accordingly, it makes sense to use the terminology
\emph{critical value of index $k$} to indicate a critical value
that is the image of a critical point of index $k$.
\begin{lem}\label{relhom}
Let $f\in \F^0\subset\F(\M,\R)$, and let $a,b\in\R$, $a<b$, be
regular values of $f$. If there exists $\overline{k}\in\Z$ such
that $H_{\overline{k}}(f^b,f^a)\neq 0$, then $[a,b]$ contains at
least one critical value of index $\overline{k}$.
\end{lem}
\begin{proof}

From
\emph{(St.~1)}, the absence of critical values in $[a,b]$ implies
that the homomorphism induced by inclusion $\iota_k:H_k(f^a)\to
H_k(f^b)$ is an isomorphism for each $k\in\Z$. Consequently, by
using the long exact sequence of the pair:
\begin{eqnarray*}
\begin{array}{ccccccccccccc}
\cdots\!\!\!&\!\!\!\longrightarrow\!\!\!&\!\!\!
H_{k}(f^a)\!\!\!&\!\!\!\stackrel{i_k}{\longrightarrow}\!\!\!&\!\!\!
H_k(f^b)\!\!\!&\!\!\!\stackrel{j_k}{\longrightarrow}\!\!\!&\!\!\!
H_k(f^b,f^a)\!\!\!&\!\!\!\stackrel{\partial_{k}}{\longrightarrow}\!\!\!&\!\!\!
H_{k-1}(f^a)\!\!\!&\!\!\!\stackrel{i_{k-1}}{\longrightarrow}\!\!\!&\!\!\!
H_{k-1}(f^b)\!\!\!&\!\!\!\longrightarrow\!\!\!&\!\!\! \cdots,
\end{array}
\end{eqnarray*}
it is easily seen that, for every $k\in\Z$, the surjectivity of
$i_k$ and the injectivity of $i_{k-1}$ imply the triviality of
$H_k(f^b,f^a)$. This proves that if there exists
$\overline{k}\in\Z$ such that $H_{\overline{k}}(f^b,f^a)\neq 0$,
then $[a,b]$ contains at least one critical value of $f$. That the
index of at least one of the critical values of $f$ contained in
$[a,b]$ is exactly $\overline{k}$ is consequent to the
sub-additivity property of the rank of the relative homology
groups and to \emph{(St.~2)}. In fact, let $c_1,\ldots,c_m$ be the
critical values of $f$ belonging to $[a,b]$, and let
$s_0,\ldots,s_{m}$ be $m+1$ regular values such that
$a=s_0<c_1<s_1<c_{2}<\ldots<s_{m-1}<c_m<s_{m}=b$. Since it holds
that $\rank H_{\overline{k}}(f^b,f^a)\leq
\underset{i=1}{\overset{m}\sum} \rank
H_{\overline{k}}(f^{s_i},f^{s_{i-1}})$, and by hypothesis $\rank
H_{\overline{k}}(f^b,f^a)\geq 1$, there exists at least one index
$i\in\{1,\ldots,m\}$ such that
$H_{\overline{k}}(f^{s_i},f^{s_{i-1}})\neq 0$. Now, applying
\emph{(St.~2)} with $a$ replaced by $s_{i-1}$ and $b$ replaced by
$s_i$, we deduce that $c_i$ is a critical value of $f$ of index
$\overline{k}$.
\end{proof}

The above statements \emph{(St.~1-2)}, Lemma \ref{relhom},
together with the following lemma, that is a reformulation of
Lemma 4.1 in \cite{MaPr75}, provide the tools for proving the
stability of critical values under small function perturbations
(Theorem \ref{mainres}).
\begin{lem}\label{MarinoProdi}
Let $X_1, X_2, X_3, X_1', X_2', X_3'$ be topological spaces such
that $X_1 \subseteq X_2 \subseteq X_3 \subseteq X_1' \subseteq
X_2' \subseteq X_3'.$ Let $H_k(X_3,X_1)=0$, $H_k(X_3',X_1')=0$ for
every $k \in \Z$. Then the homomorphism induced by inclusion
$H_k(X_1',X_1)\to H_k(X_2',X_2)$ is injective for every $k \in
\Z$.
\end{lem}

\begin{theorem}[Stability of critical values]\label{mainres}
Let $f\in \F^0\subset\F(\M,\R)$ and let $c$ be a critical value of
index $\overline{k}$ of $f$. Then there exists a real number
$\delta(f,c)>0$ such that each $g\in \F^0$ verifying
${\|f-g\|}_{C^0} \leq\delta$, $0\leq \delta \leq \delta(\f,c)$,
admits at least one critical value of index $\overline{k}$ in
$[c-\delta, c+ \delta]$.
\end{theorem}

\begin{proof}
Since $f$ is Morse, we can choose a real number $\delta(f,c)>0$
such that $[c-3\cdot\delta(\f,c), c+3\cdot\delta(\f,c)]$ does not
contain any critical value of $f$ besides $c$. Let $0\le \delta
\leq \delta(\f,c)$, and let $g$ be a simple Morse function such
that ${\|f-g\|}_{C^0} \leq\delta$. If $\delta=0$, then the claim
immediately follows. Let $\delta>0$. Then, for every $n \in \N$,
$$f^{c-\delta\cdot\frac{2n+1}{n}}\subseteq g^{c-\delta\cdot\frac{n+1}{n}}\subseteq f^{c-\delta/n}\subseteq f^{c+\delta/n}\subseteq g^{c+\delta\cdot\frac{n+1}{n}}\subseteq f^{c+\delta\cdot\frac{2n+1}{n}}.$$
Since $[c-\delta\cdot\frac{2n+1}{n},c-\delta/n]$ and
$[c+\delta/n,c+\delta\cdot\frac{2n+1}{n}]$ do not contain any
critical value of $f$ for every $n\in\N$, both
$H_k(f^{c-\delta/n},f^{c-\delta\cdot\frac{2n+1}{n}})$ and
$H_k(f^{c+\delta\cdot\frac{2n+1}{n}},f^{c+\delta/n})$ are trivial
for every $k \in \Z$, and $n\in\N$. Consequently, from Lemma
\ref{MarinoProdi}, the homomorphism induced by inclusion
$H_k(f^{c+\delta/n},f^{c-\delta\cdot\frac{2n+1}{n}})\to
H_k(g^{c+\delta\cdot\frac{n+1}{n}},g^{c-\delta\cdot\frac{n+1}{n}})$
is injective for each $k \in \Z$, and $n\in\N$. Moreover, since,
for every $n\in\N$, $[c-\delta\cdot\frac{2n+1}{n},c+\delta/n]$
contains $c$, that is a critical value of index $\overline{k}$ of
$f$, from \emph{(St.~2)}, it holds that
$H_{\overline{k}}(f^{c+\delta/n},f^{c-\delta\cdot\frac{2n+1}{n}})\neq
0$ for every $n\in\N$. This fact, together with the injectivity of
the above map, implies that also
$H_{\overline{k}}(g^{c+\delta\cdot\frac{n+1}{n}},g^{c-\delta\cdot\frac{n+1}{n}})\neq
0$ for every $n\in\N$. So, by Lemma \ref{relhom}, for every
$n\in\N$, there exists at least one critical value $c_n'$ of index
$\overline{k}$ of $g$ with
$c_n'\in(c-\delta\cdot\frac{n+1}{n},c+\delta\cdot\frac{n+1}{n})$.
By contradiction, let us suppose that $[c-\delta,c+\delta]$
contains no critical values of index $\overline{k}$ of $g$. Then,
since $g$ is Morse, there would exist a sufficiently small real
number $\varepsilon>0$ such that
$(c-\delta-\varepsilon,c+\delta+\varepsilon)$ does not contain
critical values of index $\overline{k}$ of $g$ either, giving an
absurd.
\end{proof}

We now prove the local stability of labelled Reeb graphs of closed
curves. We need a lemma that holds for manifolds of arbitrary
dimension. The global stability will be exposed in the next
section.

\begin{lem}\label{localM}
Let $f \in \F^0\subset\F(\M,\R)$. Then there exists a positive
real number $\delta(f)$ such that, for every $\delta$,
$0\le\delta\le\delta(f)$, and for every $g\in\F^0$, with
${\|f-g\|}_{C^2} \leq \delta$, an edge and vertices order
preserving bijection $\Phi:V(\Gamma_{f})\to V(\Gamma_{g})$ exists
for which $\underset{v\in V(\Gamma_{f})}\max
|f_{_|}(v)-g_{_|}(\Phi(v))|\le\delta.$
\end{lem}
\begin{proof}
Let $p_1,\ldots,p_n$ be the critical points of $f$, and
$c_1,\ldots,c_n$ the respective critical values, with
$c_i<c_{i+1}$ for $i=1\ldots,n-1$. Since $\F^0$ is open in $\F(\M,
\R)$, endowed with the $C^2$ topology, there always exists a
sufficiently small $\delta(f)>0$, such that the closed ball with
center $f$ and radius $\delta(f)$, $\overline{B_2(f,\delta(f))}$,
is contained in $\F^0$. Moreover, $\delta(f)$ can be chosen so
small that,  for every $i=1,\ldots,n-1$, the intervals
$[c_i-\delta(f),c_i+\delta(f)]$ and
$[c_{i+1}-\delta(f),c_{i+1}+\delta(f)]$ are disjoint.

Fixed such a $\delta(f)$, for every real number $\delta$, with
$0\le\delta\le\delta(f)$, and for every $g\in\F^0$ such that
${\|f-g\|}_{C^2} \leq \delta$, $f$ and $g$ belong to the same
arcwise connected component of $\F^0$ endowed with the
$C^{\infty}$ topology, and, therefore, are topologically
equivalent functions. Consequently, there exists an edge and
vertices order preserving bijection $\Phi:V(\Gamma_{f})\to
V(\Gamma_{g})$ (see Subsection \ref{reebsection}). Let us prove
that $\Phi$ is such that $\underset{v\in V(\Gamma_{f})}\max
|f_{_|}(v)-g_{_|}(\Phi(v))|\le\delta.$ Since $f$ and $g$ are
topologically equivalent, it follows that $g$ has exactly $n$
critical points, $p_1',\ldots,p_n'$. Let
$c_1'=g(p_1'),\ldots,c_n'=g(p_n')$. We can assume $c_i'<c_{i+1}'$,
for $i=1,\ldots,n-1$. The assumption ${\|f-g\|}_{C^2} \leq \delta$
implies that ${\|f-g\|}_{C^0} \leq \delta$. Therefore, by the
previous Theorem \ref{mainres}, for every critical value $c_i$ of
$f$, there exists at least one critical value of $g$ of the same
index of $c_i$ belonging to $[c_i-\delta,c_i+\delta]$. Moreover,
since $[c_i-\delta,c_i+\delta]\cap
[c_{i+1}-\delta,c_{i+1}+\delta]=\emptyset$ for every
$i=1,\ldots,n-1$, it follows that $c_i'\in
[c_i-\delta,c_i+\delta]$ for every $i=1,\ldots,n$. Hence, since
$\Phi$ preserves the order of the vertices, necessarily
$\Phi(p_i)=p_i'$, yielding that $\underset{v\in V(\Gamma_{f})}\max
|f_{_|}(v)-g_{_|}(\Phi(v))|=\underset{p_i\in K(f)}\max
|f_{_|}(p_i)-g_{_|}(\Phi(p_i))|=\underset{1\le i\le n}\max
|c_i-c_i'|\le\delta.$
\end{proof}

\begin{theorem}[Local stability]\label{local}
Let $f \in \F^0\subset\F(S^1,\R)$. Then there exists a positive
real number $\delta(f)$ such that, for every $\delta$,
$0\le\delta\le\delta(f)$, and for every $g\in\F^0$, with
${\|f-g\|}_{C^2} \leq \delta$, it holds that
$d((\Gamma_{f},f_{_|}),(\Gamma_{g},g_{_|}))\le\delta.$
\end{theorem}

\begin{proof}
By Lemma \ref{localM}, an edge and vertices order preserving
bijection $\Phi:V(\Gamma_{f})\to V(\Gamma_{g})$ exists for which
$\underset{v\in V(\Gamma_{f})}\max
|f_{_|}(v)-g_{_|}(\Phi(v))|\le\delta.$ Necessarily $\Phi$ takes
minima into minima and maxima into maxima. Therefore,
$(\Gamma_{f},g_{_|}\circ\Phi)=T(\Gamma_{f},f_{_|})$, with $T$ an
elementary deformation of type (R), relabelling vertices of
$V(\Gamma_{f})$, having cost $c(T)=\underset{v\in
V(\Gamma_{f})}\max |f_{_|}(v)-g_{_|}(\Phi(v))|\le\delta$.
Moreover, let us observe that $(\Gamma_{f},g_{_|}\circ\Phi)$ is
isomorphic to $(\Gamma_{g},g_{_|})$ as labelled Reeb graph (see
Definition \ref{isolabel}). Thus,
$d((\Gamma_{f},f_{_|}),(\Gamma_{g},g_{_|}))=d((\Gamma_{f},f_{_|}),(\Gamma_{f},g_{_|}\circ\Phi))=\underset{T\in
\mathcal{T}((\Gamma_{f}, f_{_|}),(\Gamma_{g},g_{_|}))}\inf c(T)\le
\delta.$
\end{proof}

\section{Global stability}\label{globalstab}
This section is devoted to proving that Reeb graphs of closed
curves are stable under arbitrary function perturbations. More
precisely, it will be shown that arbitrary changes in simple Morse
functions imply smaller changes in the editing distance between
Reeb graphs. The proof is by steps: the following Proposition
\ref{pathstability} shows such a stability property when the
functions defined on $S^1$ belong to the same arcwise connected
component of $\F^0$; Proposition \ref{pointonF1} proves the same
result in the case that the linear convex combination of two
simple Morse functions traverses the stratum $\F^1$ at most in one
point; Theorem \ref{global} extends the result to two arbitrary
functions in $\F^0$.

\begin{prop}\label{pathstability}
Let $f, g \in \F^0$ and let us consider the path $\h: [0,1]\to
\F(S^1,\R)$ defined by $\h(\lambda)=(1-\lambda)f+\lambda g$. If
$\h(\lambda)  \in \F^0$ for every $\lambda \in [0,1]$, then
$d((\Gamma_{f},f_{_|}),(\Gamma_{g},g_{_|}))\le {\|f-g\|}_{C^2}.$
\end{prop}
\begin{proof}
Let $\delta(h(\lambda))>0$ be the fixed real number playing the
same role of $\delta(f)$ in Theorem \ref{local}, after replacing
$f$ by $h(\lambda)$. For conciseness, let us denote it by
$\delta(\lambda)$, and ${\|f-g\|}_{C^2}$ by $a$. If $a=0$, the
claim trivially follows. If $a>0$, let $C$ be the open covering of
$[0,1]$ constituted of open intervals
$I_{\lambda}=\left(\lambda-\frac{\delta(\lambda)}{2a},\lambda+\frac{\delta(\lambda)}{2a}\right)$.
Let $C'$ be a finite minimal (i.e. such that, for every $i$,
$I_{\lambda_i}\nsubseteq \underset{j\neq i}\bigcup I_{\lambda_j}$)
sub-covering of $C$, with $\lambda_1<\lambda_2<\ldots<\lambda_n$
the middle points of its intervals. Since $C'$ is minimal, for
every $i\in\{1,\ldots, n-1\}$, $I_{\lambda_i}\cap
I_{\lambda_{i+1}}$ is non-empty. This implies that
{\setlength\arraycolsep{2pt}\begin{eqnarray}\label{first}
\lambda_{i+1} -
\lambda_i&<&\frac{\delta(\lambda_i)}{2a}+\frac{\delta(\lambda_{i+1})}{2a}\le\frac{\max\{\delta(\lambda_i),\delta(\lambda_{i+1})\}}{a}.
\end{eqnarray}}
Moreover, by the definition of $\h$ and the linearity of
derivatives, it can be deduced that
{\setlength\arraycolsep{2pt}\begin{eqnarray}\label{second}
{\|h(\lambda_{i+1})-h(\lambda_{i})\|}_{C^2}&=&(\lambda_{i+1}
-\lambda_{i})\cdot{\|f-g\|}_{C^2}.
\end{eqnarray}}
Now, substituting (\ref{first}) in (\ref{second}), we obtain
\begin{eqnarray*}
{\|h(\lambda_{i+1})-h(\lambda_{i})\|}_{C^2}<\frac{\max\{\delta(\lambda_i),\delta(\lambda_{i+1})\}}{a}\cdot{\|f-g\|}_{C^2}=\max\{\delta(\lambda_i),\delta(\lambda_{i+1})\}.
\end{eqnarray*}

Let $(\Gamma_{\h(\lambda_{j})}, \h(\lambda_{j})_{_|})$ be the
labelled Reeb graphs associated with $(S^1,\h(\lambda_{j}))$,
$j=1, \dots,n$. Let $i\in\{1,\ldots, n-1\}$. If
$\max\{\delta(\lambda_i),\delta(\lambda_{i+1})\}=\delta(\lambda_i)$,
then using Theorem \ref{local}, with $f$ replaced by
$h(\lambda_{i})$, $g$ by $h(\lambda_{i+1})$ and $\delta$ by
${\|h(\lambda_{i+1})-h(\lambda_{i})\|}_{C^2}$, it holds that
{\setlength\arraycolsep{2pt}\begin{eqnarray}\label{four}
d((\Gamma_{h(\lambda_{i})},
h(\lambda_{i})_{_|}),(\Gamma_{h(\lambda_{i+1})},
h(\lambda_{i+1})_{_|}))&\le&{\|h(\lambda_{i+1})-h(\lambda_{i})\|}_{C^2}.
\end{eqnarray}}
The same inequality holds when
$\max\{\delta(\lambda_i),\delta(\lambda_{i+1})\}=\delta(\lambda_{i+1})$,
as can be analogously checked.

Now, setting $\lambda_0=0$, $\lambda_{n+1}=1$, it can be verified
that (\ref{four}) also holds for $i=0, n$. Consequently, since
$\Gamma_{f}=\Gamma_{h(\lambda_{0})}$, and
$\Gamma_{g}=\Gamma_{h(\lambda_{n+1})}$, we have
{\setlength\arraycolsep{2pt}\begin{eqnarray*}
d((\Gamma_{f},f_{_|}),(\Gamma_{g},g_{_|}))&\le&\underset{i=0}{\overset{n}\sum}d((\Gamma_{h(\lambda_{i})},
h(\lambda_{i})_{_|}),(\Gamma_{h(\lambda_{i+1})},
h(\lambda_{i+1})_{_|}))\le\underset{i=0}{\overset{n}\sum}{\|h(\lambda_{i+1})-h(\lambda_{i})\|}_{C^2}\\
&=&\underset{i=0}{\overset{n}\sum}(\lambda_{i+1}
-\lambda_{i})\cdot{\|f-g\|}_{C^2}={\|f-g\|}_{C^2},
\end{eqnarray*}}
where the first inequality is due to the triangular inequality,
the second one to (\ref{four}), the first equality holds because
of (\ref{second}), the second one because
$\underset{i=0}{\overset{n}\sum}(\lambda_{i+1} -\lambda_{i})=1.$
\end{proof}

\begin{prop}\label{pointonF1}
Let $f, g \in \F^0$ and let us consider the path $\h: [0,1]\to
\F(S^1,\R)$ defined by $\h(\lambda)=(1-\lambda)f+\lambda g$. If
$\h(\lambda)  \in \F^0$ for every $\lambda \in
[0,1]\setminus\{\overline{\lambda}\}$, with $0<\overline{\lambda}<1$, and
$\h$  transversely intersects $\F^1$ at $\overline{\lambda}$, then
$d((\Gamma_{f},f_{_|}),(\Gamma_{g},g_{_|}))\le {\|f-g\|}_{C^2}.$
\end{prop}
\begin{proof}
We begin proving the following claim.\\
\noindent{\bf Claim.} For every $\delta>0$ there exist two real
numbers $\lambda',\lambda''\in [0,1]$, with $\lambda'<\overline{\lambda}<\lambda''$, such that
$d((\Gamma_{h(\lambda')},h(\lambda')_{_|}),(\Gamma_{\h(\lambda'')},h(\lambda'')_{_|}))\le\delta.$

To prove this claim, let us first assume that
$\h(\overline{\lambda})$ belongs to $\F^1_{\alpha}$. To simplify
the notation, we denote $\h(\overline{\lambda})$ simply by
$\overline{h}$. Let $\overline{p}$ be the sole degenerate critical
point for $\overline{h}$. It is well known that there exists  a
suitable local coordinate system $x$ around $\overline{p}$ in
which the canonical expression of $\overline{h}$ is
$\overline{h}=\overline{h}(\overline{p})+x^3$ (see Subsection
\ref{stratification} and Figure \ref{morseperturb} $(a)$ with
$\overline{h}$ replaced by $f$).

Let us take a smooth function $\omega: S^1\to \R$ whose support is
contained in the coordinate chart around $\overline{p}$  in which
$\overline{h}=\overline{h}(\overline{p})+x^3$; moreover, let us
assume that $\omega$ is equal to $1$ in a neighborhood of
$\overline{p}$, and decreases moving from $\overline{p}$. Let us
consider  the family of smooth functions $\overline{h}_t$ obtained
by locally modifying $\overline{h}$ near $\overline{p}$ as
follows: $\overline{h}_t=\overline{h}+t\cdot \omega\cdot x$. There
exists $\overline{t}> 0$ sufficiently small such that $(i)$ for
$0<t\le \overline{t}$, $\overline{h}_t$ has no critical points in
the support of $\omega$ and is equal to $\overline{h}$ everywhere
else (see Figure \ref{morseperturb} $(a)$ with $\overline{h}_t$
replaced by $\widetilde{f}_2$), and $(ii)$ for $-\overline{t}\le
t<0$, $\overline{h}_t$ has exactly two critical points in the
support of $\omega$ whose values difference tends to vanish as $t$
tends to 0, and $\overline{h}_t$ is equal to $\overline{h}$
everywhere else (see \cite{Cerf70} and Figure \ref{morseperturb}
$(a)$ with $\overline{h}_t$ replaced by $\widetilde{f}_1$).

Since $\overline{h}_t$ is a universal deformation of
$\overline{h}=h(\overline{\lambda})$, and $h$ intersect $\F^1$
transversely at $\overline{\lambda}$, either the maps $h(\lambda)$
with $\lambda<\overline{\lambda}$ are topologically equivalent to
$h_t$ with $t>0$ or to $h_t$ with $t<0$ (cf. \cite{Cerf70, Ma82,
Se72}). Analogously for the maps $h(\lambda)$ with
$\lambda>\overline{\lambda}$. Let us assume that $h(\lambda)$ is
topologically equivalent to $h_t$ with $t<0$ when
$\lambda<\overline{\lambda}$, while $h(\lambda)$ is topologically
equivalent to $h_t$ with $t>0$ when $\lambda>\overline{\lambda}$.
Hence, for every $\delta>0$, there exist $\lambda'$, with $0\le
\lambda'<\overline{\lambda}$, and $\lambda''$, with
$\overline{\lambda}<\lambda''\le 1$, such that $h(\lambda')$ and
$h(\lambda'')$ have the same critical points, with the same
values, except for two critical points of $h(\lambda')$, whose
values difference is smaller than $\delta$, that are non-critical
for $h(\lambda'')$. Therefore,
$(\Gamma_{h(\lambda')},h(\lambda')_{_|})$ can be transformed into
$(\Gamma_{h(\lambda'')},h(\lambda'')_{_|})$ by an elementary
deformation of type (D) whose cost is not greater than $\delta$.
In the case when $h(\lambda)$ is topologically equivalent to $h_t$
with $t>0$ when $\lambda<\overline{\lambda}$, while $h(\lambda)$
is topologically equivalent to $h_t$ with $t<0$ when
$\lambda>\overline{\lambda}$, the claim can be proved similarly,
applying an elementary deformation of type (B).

Let us now prove the claim when
$\overline{h}=h(\overline{\lambda})$ belongs to $\F^1_{\beta}$.
Let us denote by $\overline{p}$ and $\overline{q}$ the critical
points of $\overline{h}$ such that
$\overline{h}(\overline{p})=\overline{h}(\overline{q})$. Since
$\overline{p}$ is non-degenerate there exists  a suitable local
coordinate system $x$ around $\overline{p}$ in which the canonical
expression of $\overline{h}$ is
$\overline{h}=\overline{h}(\overline{p})+x^2$ (see Figure
\ref{morseperturb} $(b)$ with $\overline{h}$ replaced by $f$). Let
us take $\omega$ as before, whose support is contained in such a
coordinate chart. Let us locally modify $\overline{h}$ near
$\overline{p}$ as follows: $\overline{h}_t=\overline{h}+t\cdot
\omega$. There exists $\overline{t}> 0$ sufficiently small such
that for $|t|\le \overline{t}$, $\overline{h}_t$ has exactly the
same critical points  as $\overline{h}$. As for critical values,
they are the same as well, apart from the value taken at
$\overline{p}$:
$\overline{h}_t(\overline{p})<\overline{h}(\overline{p})$, for
$-\overline{t}\le t<0$ (see Figure \ref{morseperturb} $(b)$ with
$\overline{h}_t$ replaced by $\widetilde{f}_1$),  while
$\overline{h}_t(\overline{p})>\overline{h}(\overline{p})$, for
$0<t\le\overline{t}$ (see Figure \ref{morseperturb} $(b)$ with
$\overline{h}_t$ replaced by $\widetilde{f}_2$), and
$\overline{h}_t(\overline{p})$ tends to
$\overline{h}(\overline{p})$ as $t$ tends to $0$ (cf.
\cite{Cerf70}). Since $\overline{h}_t$ is a universal deformation
of $\overline{h}=h(\overline{\lambda})$, and $h$ intersect $\F^1$
transversely at $\overline{\lambda}$, we deduce that for every
$\delta>0$ there exist $\lambda'$, with $0\le
\lambda'<\overline{\lambda}$ and $\lambda''$, with
$\overline{\lambda}<\lambda''\le 1$, such that
$(\Gamma_{\h(\lambda')},\h(\lambda')_{_|})$ can be transformed
into $(\Gamma_{\h(\lambda'')},\h(\lambda'')_{_|})$ by an
elementary deformation of type (R) whose cost is not greater than
$\delta$. Therefore  the initial claim is proved.

Let us now estimate $d((\Gamma_{f},f_{_|}),(\Gamma_{g},g_{_|}))$.
By the claim, for every $\delta >0$, there exist
$0<\lambda'<\lambda ''<1$ such that, applying the triangular
inequality, {\setlength\arraycolsep{2pt}\begin{eqnarray*}
d((\Gamma_{f},f_{_|}),(\Gamma_{g},g_{_|}))&\le&d((\Gamma_{f},f_{_|}),(\Gamma_{h(\lambda')},h(\lambda')_{_|}))+d((\Gamma_{h(\lambda')},h(\lambda')_{_|}),
(\Gamma_{h(\lambda'')},h(\lambda'')_{_|}))\\&&+d((\Gamma_{h(\lambda'')},h(\lambda'')_{_|}),(\Gamma_{g},g_{_|}))
\\&\le &
d((\Gamma_{f},f_{_|}),(\Gamma_{h(\lambda')},h(\lambda')_{_|}))+d((\Gamma_{h(\lambda'')},h(\lambda'')_{_|}),(\Gamma_{g},g_{_|}))+\delta.
\end{eqnarray*}}
By Proposition \ref{pathstability},
$$ d((\Gamma_{f},f_{_|}),(\Gamma_{h(\lambda')},h(\lambda')_{_|}))\le \|f-h(\lambda')\|_{C^2}=\lambda'\cdot \|f-g\|_{C^2},$$ and
$$ d((\Gamma_{h(\lambda'')},h(\lambda'')_{_|}),(\Gamma_{g},g_{_|}))\le \|h(\lambda'')-g\|_{C^2}=(1-\lambda'')\cdot \|f-g\|_{C^2}.$$
Hence, $d((\Gamma_{f},f_{_|}),(\Gamma_{g},g_{_|}))\le
\|f-g\|_{C^2}+\delta$, yielding the conclusion by the
arbitrariness of $\delta$.
\end{proof}

\begin{theorem}[Global stability]\label{global}
Let $f, g \in \F^0$. Then
$d((\Gamma_{f},f_{_|}),(\Gamma_{g},g_{_|}))\le {\|f-g\|}_{C^2}.$
\end{theorem}
\begin{proof}
For every sufficiently small $\delta>0$ such that
$B_2(f,\delta),B_2(g,\delta)\subset \F^0$, there exist
$\widehat{f}\in B_2(f,\delta)$ and $\widehat{g} \in B_2(g,\delta)$
such that the path $h:[0,1]\to \F(S^1,\R)$, with
$h(\lambda)=(1-\lambda)\widehat{f}+\lambda\widehat{g}$, belongs to
$\F^0$ for every $\lambda \in [0,1]$, except for at most a finite
number $n$ of values  $0<\mu_1<\mu_2<\ldots<\mu_n<1$ at which $h$
transversely intersects $\F^1$. If $n=0$ ($n=1$, respectively),
then the claim immediately follows from Proposition
\ref{pathstability} (Proposition \ref{pointonF1}, respectively).
If $n>1$, let $0<\lambda_1<\lambda_2<\ldots<\lambda_{2n-1}<1$,
with $\lambda_{2i-1}=\mu_i$ for $i=1,\ldots,n$. Then
$h(\lambda_{2i-1})\in \F^1$ for $i=1,\ldots,n$,
$h(\lambda_{2i})\in \F^0$ for $i=1,\dots, n-1$. Set
$\lambda_{0}=0$ so that $\widehat{f}=h(\lambda_{0})$, and
$\lambda_{2n}=1$ so that $\widehat{g}=h(\lambda_{2n})$ (a
schematization of this path can be visualized in Figure
\ref{stabilityProof}). Then, by Proposition \ref{pointonF1}, we
have
$$d((\Gamma_{h(\lambda_{2i})},h(\lambda_{2i})_{_|}),
(\Gamma_{h(\lambda_{2i+2})},h(\lambda_{2i+2})_{_|}))\le
{\|h(\lambda_{2i})-h(\lambda_{2i+2})\|}_{C^2}$$ for every
$i=0,\ldots,n-1.$ Therefore
{\setlength\arraycolsep{2pt}\begin{eqnarray*}
d((\Gamma_{\widehat{f}},\widehat{f}_{_|}),(\Gamma_{\widehat{g}},\widehat{g}_{_|}))&\le&\underset{i=0}{\overset{n-1}\sum}d((\Gamma_{h(\lambda_{2i})},h(\lambda_{2i})_{_|}),
(\Gamma_{h(\lambda_{2i+2})},h(\lambda_{2i+2})_{_|}))\\&\le&
\underset{i=0}{\overset{n-1}\sum}{\|h(\lambda_{2i})-h(\lambda_{2i+2})\|}_{C^2}\le{\|\widehat{f}-\widehat{g}\|}_{C^2}.\\
\end{eqnarray*}}
Then, recalling that $\widehat{f}\in B_2(f,\delta)$ means
${\|\widehat{f}-f\|}_{C^2}\le\delta$, and
$B_2(f,\delta)\subset\F^0$ implies that $(1-\lambda)f+\lambda
\widehat{f}\in\F^0$ for every $\lambda\in[0,1]$, we can apply
Proposition \ref{pathstability} to state that
$d((\Gamma_{f},f_{_|}),(\Gamma_{\widehat{f}},\widehat{f}_{_|}))\le\delta$.
It is analogous for $g$ and $\widehat{g}$. Thus, from the
triangular inequality, we have
{\setlength\arraycolsep{2pt}\begin{eqnarray*}
d((\Gamma_{f},f_{_|}),(\Gamma_{g},g_{_|}))&\le&
d((\Gamma_{f},f_{_|}),(\Gamma_{\widehat{f}},\widehat{f}_{_|}))+d((\Gamma_{\widehat{f}},\widehat{f}_{_|}),(\Gamma_{\widehat{g}},\widehat{g}_{_|}))+d((\Gamma_{\widehat{g}},\widehat{g}_{_|}),(\Gamma_{g},g_{_|}))\\&\le&
2\delta + {\|\widehat{f}-\widehat{g}\|}_{C^2}.\end{eqnarray*}}
Now, since by the triangular inequality,
${\|\widehat{f}-\widehat{g}\|}_{C^2}\le{\|\widehat{f}-f\|}_{C^2}+{\|f-g\|}_{C^2}+{\|g-\widehat{g}\|}_{C^2}$,
with ${\|\widehat{f}-f\|}_{C^2}\le\delta$, and
${\|g-\widehat{g}\|}_{C^2}\le\delta$, it follows that
$d((\Gamma_{f},f_{_|}),(\Gamma_{g},g_{_|}))\le
4\delta+{\|f-g\|}_{C^2}$. Finally, because of the arbitrariness of
$\delta$, we can let $\delta$ tend to zero and obtain the claim.
\end{proof}

\begin{figure}[htbp]
\psfrag{f=q0}{$\widehat{f}=h(\lambda_{0})$}\psfrag{qn=g}{$h(\lambda_{2n})=\widehat{g}$}\psfrag{in}{$\in$}\psfrag{F0}{$\F^0$}\psfrag{F1}{$\F^1$}\psfrag{p3}{$v_3$}\psfrag{v4}{$v_4$}
\psfrag{v5}{$v_5$}\psfrag{v6}{$v_6$}\psfrag{v7}{$v_7$}\psfrag{v8}{$v_8$}
\psfrag{p1}{$h(\mu_{1})$}\psfrag{q1}{$h(\lambda_{1})$}\psfrag{p2}{$h(\mu_{2})$}\psfrag{q3}{$h(\lambda_{3})$}
\psfrag{=}{$=$}\psfrag{q2}{$h(\lambda_{2})$}\psfrag{p3}{$h(\mu_{3})$}\psfrag{q4}{$h(\lambda_{4})$}
\psfrag{q5}{$h(\lambda_{5})$}\psfrag{qn3}{$h(\lambda_{2n-3})$}\psfrag{qn2}{$h(\lambda_{2n-2})$}\psfrag{qn1}{$h(\lambda_{2n-1})$}
\psfrag{pn1}{$h(\mu_{n-1})$}\psfrag{pn}{$h(\mu_{n})$}
\includegraphics[height=3cm]{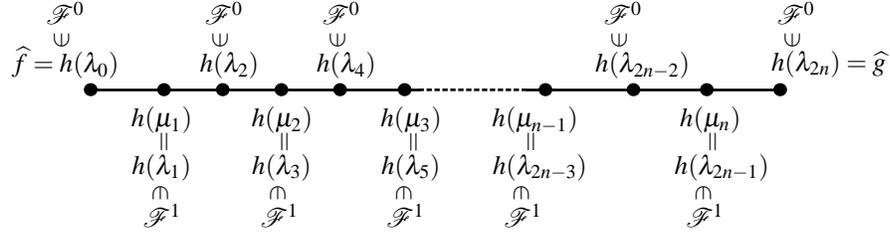}
\caption{\footnotesize{The linear path used in the proof of
Theorem \ref{global}.}}\label{stabilityProof}
\end{figure}

\section{Discussion}
In this paper, we have considered Reeb graphs of curves and have
shown that they stably represent topological properties of smooth
functions. Precisely, we have constructed an editing distance
between Reeb graphs of closed curves endowed with smooth functions
$f$ and $g$, that is bounded from below by the natural
pseudo-distance between $(S^1,f)$ and $(S^1,g)$, and from above by
the $C^2$-norm of $f-g$.

This paper is meant as a first step toward the study of stability
of Reeb graphs of surfaces. While the general technique we use to
prove our main result, as well as many intermediate results, could
be easily generalized to surfaces, the definition of the editing
distance would need to be appropriately modified. This requires us
to classify the possible degeneracies of Reeb graphs of surfaces.
Moreover, our proof of the metric properties of the editing
distance exploits some particular properties of curves that are no
longer valid for surfaces.

Furthermore, other shape descriptors consisting of graphs
constructed out of Morse theory, such as the Morse Connection
Graph introduced in \cite{CoAlZi05} and further developed in
\cite{AlCo*07}, could possibly benefit of some of the  results
proved  in this paper.

However, some questions remain unanswered also in the case of
curves. In the examples shown in this paper, the editing distance
coincides with the natural pseudo-distance. Is this always the
case? Moreover,  looking at the analogous results proved in
\cite{CoEdHa07, dAFrLa} about the stability of persistent homology
groups, another shape descriptor used both in computer vision and
computer graphics for shape comparison,  we may notice that the
$C^0$-norm rather than the $C^2$-norm is used to evaluate function
changes. So another open question, strictly related to the
previous one, is whether it would be possible to improve our
result in this sense. Other open questions are concerned with
applications of the Main Result (Theorem \ref{global}) to measure
shape dissimilarity coping well with noisy data. On one hand, the
result ensures the stability of Reeb graphs against noise, while,
on the other, we may wonder how likely it is that noise
encountered in real data is small with respect to the $C^2$-norm.
Indeed, it is easy to conceive examples where perturbations that
could be seen as noise do not correspond to a small value of the
$C^2$-norm. For example, the functions represented in Figure
\ref{noise} belong to a sequence of functions $(f_n)$ all having
the same $C^2$-norm although they tend to $0$ with respect to the
$C^0$-norm. However, one could argue that in a discrete setting,
at a fixed resolution, sequences of functions as in Figure
\ref{noise} cannot be found. Moreover, this problem would be
overcome if the editing distance coincides with the natural
pseudo-distance.

\begin{figure}[htbp]
\psfrag{f1}{$f_1$}\psfrag{f2}{$f_2$}\psfrag{f3}{$f_3$}\psfrag{1}{$1$}\psfrag{0}{$0$}\psfrag{2}{$2$}\psfrag{12}{$1/2$}\psfrag{13}{$1/3$}\psfrag{34}{$3/4$}
\includegraphics[height=5cm]{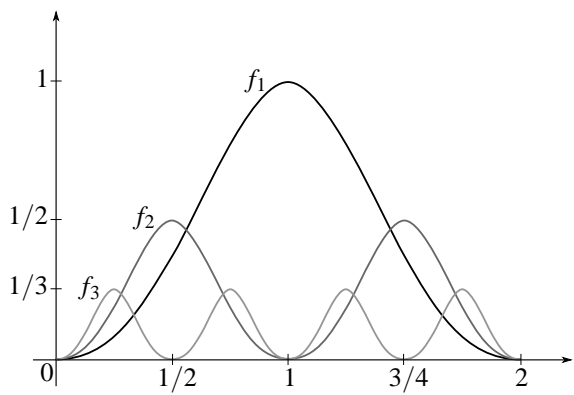}
\caption{\footnotesize{The graphs of three functions having the same $C^2$-norm.}}\label{noise}
\end{figure}

\bibliographystyle{amsplain}
\bibliography{biblio}

\end{document}